\documentclass[journal]{IEEEtran}
\usepackage{amsthm}
\usepackage{graphicx}
\usepackage{tikz}
\usepackage{multirow}
\usepackage{amsmath}
\usepackage{adjustbox}
\newcommand*\circled[1]{\tikz[baseline=(char.base)]{
            \node[shape=circle,draw,inner sep=1pt] (char) {#1};}}

\begin{document}

\title{Enabling Seamless Data Security, Consensus, and Trading in Vehicular Networks}

\author{Emanuel~Vieira*,~\IEEEmembership{Member,~IEEE,}
        João~Almeida,~\IEEEmembership{Member,~IEEE,}
        Joaquim~Ferreira,~\IEEEmembership{Senior Member,~IEEE,}
        and~Paulo~C.~Bartolomeu,~\IEEEmembership{Senior Member,~IEEE}% <-

\thanks{E. Vieira and P. C. Bartolomeu are with the Instituto de Telecomunicações and Departamento de Eletrónica, Telecomunicações e Informática of the Universidade de Aveiro, Aveiro, Portugal}% <-this % stops a space
\thanks{J. Almeida is with the Instituto de Telecomunicações, Universidade de Aveiro, Aveiro, Portugal}
\thanks{J. Ferreira is with the Instituto de Telecomunicações and Escola Superior de Tecnologia e Gestão de Águeda of the Universidade de Aveiro, Aveiro, Portugal}
\thanks{*Corresponding author e-mail: vieira.e@ua.pt}
}

\markboth{Journal of \LaTeX\ Class Files,~Vol.~14, No.~8, August~2015}%
{Shell \MakeLowercase{\textit{et al.}}: Bare Demo of IEEEtran.cls for IEEE Journals}

\maketitle

\begin{abstract}
Cooperative driving is an emerging paradigm to enhance the safety and efficiency of autonomous vehicles. To ensure successful cooperation, road users must reach a consensus for making collective decisions, while recording vehicular data to analyze and address failures related to such agreements. This data has the potential to provide valuable insights into various vehicular events, while also potentially improving accountability measures.
Furthermore, vehicles may benefit from the ability to negotiate and trade services among themselves, adding value to the cooperative driving framework.
However, the majority of proposed systems aiming to ensure data security, consensus, or service trading, lack efficient and thoroughly validated mechanisms that consider the distinctive characteristics of vehicular networks.
These limitations are amplified by a dependency on the centralized support provided by the infrastructure.
Furthermore,
corresponding mechanisms must diligently address security concerns, especially regarding potential malicious or misbehaving nodes, while also considering inherent constraints of the wireless medium.
We introduce the \textit{Verifiable Event Extension} (VEE), an applicational extension designed for Intelligent Transportation System (ITS) messages. The VEE operates seamlessly with any existing standardized vehicular communications protocol, addressing crucial aspects of data security, consensus, and trading with minimal overhead. To achieve this, we employ blockchain techniques, Byzantine fault tolerance (BFT) consensus protocols, and cryptocurrency-based mechanics. To assess our proposal's feasibility and lightweight nature, we employed a hardware-in-the-loop setup for analysis. Experimental results demonstrate the viability and efficiency of the VEE extension in overcoming the challenges posed by the distributed and opportunistic nature of wireless vehicular communications.
\end{abstract}


\IEEEpeerreviewmaketitle

\section{Introduction}

\begin{table}
\caption{List of acronyms.}
\label{tab:acronyms}
\centering
\begin{tabular}{ll}
\hline

\textbf{Acronym} & \textbf{Description} \\ 
\hline

\textbf{ASN.1} & Abstract Syntax Notation One \\ 
\textbf{BFT} & Byzantine Fault Tolerance \\
\textbf{BTP} & Basic Transport Protocol \\
\textbf{BSM} & Basic Safety Message \\ 
\textbf{CAM} & Cooperative Awareness Message \\ 
\textbf{CAV} & Connected Autonomous Vehicle \\ 
\textbf{CBR} & Channel Busy Ratio \\
\textbf{CCAM} & Cooperative, Connected and Automated Mobility \\
\textbf{CCH} & Control Channel \\ 
\textbf{CPM} & Collective Perception Message \\ 
\textbf{CPU} & Central Processing Unit \\
\textbf{C-V2X} & Cellular-V2X \\ 
\textbf{D2D} & Device-to-Device \\ 
\textbf{DAG} & Directed Acyclic Graph \\ 
\textbf{DCC} & Decentralized Congestion Control \\
\textbf{DENM} & Decentralized Environmental Notification Message \\ 
\textbf{DLT} & Distributed Ledger Technology \\ 
\textbf{eCDF} & Empirical Cumulative Distribution Function \\
\textbf{ECDSA} & Elliptic Curve Digital Signature Algorithm \\
\textbf{ECIES} & Elliptic Curve Integrated Encryption Scheme \\
\textbf{EDR} & Event Data Recorder \\
\textbf{ETC} & Electronic Toll Collection \\
\textbf{ETSI} & European Telecommunications Standards Institute \\

\textbf{EV} & Electric Vehicle \\

\textbf{GN} & GeoNetworking \\
\textbf{GNSS} & Global Navigation Satellite System \\

\textbf{HiL} & Hardware-in-the-Loop \\

\textbf{ID} & IDentifier \\
\textbf{IEEE} & Institute of Electrical and Electronics Engineers \\
\textbf{IoT} & Internet of Things \\
\textbf{IPC} & Inter-Process Communication \\
\textbf{ITS} & Intelligent Transportation Systems \\

\textbf{LLC} & Logical Link Control \\
\textbf{MAC} & Medium Access Control \\
\textbf{MCM} & Maneuver Coordination Message \\

\textbf{OBU} & On-Board Unit \\
\textbf{OEM} & Original Equipment Manufacturer \\
\textbf{OMNeT++} & Objective Modular Network Testbed in C++ \\

\textbf{PBFT} & Practical Byzantine Fault Tolerance \\
\textbf{PDR} & Packet Delivery Ratio \\
\textbf{PDU} & Packet Data Unit \\
\textbf{PKI} & Public Key Infrastructure \\
\textbf{RAM} & Random Access Memory \\
\textbf{RSU} & Roadside Unit \\
\textbf{RV} & Relevant Vehicle \\

\textbf{SAE} & Society of Automotive Engineers \\

\textbf{SAEM} & Service Advertisment Essential Message \\
\textbf{SCH} & Service Channel \\
\textbf{SP} & Sub-Protocol \\
\textbf{SUMO} & Simulation of Urban Mobility \\
\textbf{SV} & Subject Vehicle \\

\textbf{TAM} & Toll Advertisement Message \\
\textbf{TV} & Target Vehicle \\
\textbf{TGS} & Traffic Generation System \\
\textbf{TUM} & Toll Usage Message \\
\textbf{TUMack} & Toll Usage Message Acknowledgement \\

\textbf{UP} & Underlying Protocol \\

\textbf{VAS} & Value-added Service \\

\textbf{VCM} & Verifiable Cooperation Message \\
\textbf{VERCO} & VERifiable COoperation \\

\textbf{VEE} & Verifiable Event Extension \\
\textbf{VEP} & Verifiable Event Protocol \\
\textbf{VE} & Verifiable Event \\
\textbf{VRU} & Vulnerable Road User \\
\textbf{VsE} & Verifiable sub-Event \\

\textbf{V2X} & Vehicle-to-everything \\
\textbf{V2V} & Vehicle-to-vehicle \\
\textbf{V2I} & Vehicle-to-infrastructure \\

\textbf{WAVE} & Wireless Access in Vehicular Environments \\
\textbf{WGS} & World Geodetic System \\
\textbf{WSA} & WAVE Service Advertisement \\
\textbf{WSMP} & WAVE Short Message Protocol \\

\hline

\end{tabular}
\end{table}

Localized vehicular communications is a key technology in the Cooperative, Connected and Automated Mobility (CCAM) paradigm. Vehicle-to-everything (V2X) wireless technologies such as ITS-G5 (802.11p) and Cellular-V2X (C-V2X) allow for vehicles to communicate among themselves in a peer-to-peer and decentralized manner without the added latency of and reliance on the infrastructure, base stations, and the backhaul network. 
The utilization of these technologies is foreseen to enhance both traffic efficiency and safety while also facilitating the deployment of non-safety-critical applications like value-added services (VAS), such as trading mechanisms and payments among road users.
However, these types of communications have limitations inherited from their wireless nature. These systems usually operate within the 5.9 GHz range and demonstrate a practical operating range of up to several hundred meters. However, packet loss issues may arise beyond this range, exacerbated by obstacles such as buildings, environmental landscape features, and even vehicles, especially in high-density traffic scenarios.
Furthermore, trust among the different actors is a central issue in vehicular networks \cite{arif19}. 
While security standards mandate that the identities road users employ in the communications among themselves are typically authenticated using cryptographic signing algorithms based on public key infrastructures (PKIs), different types of faults may still occur.
The systems (i.e., the autonomous vehicles) generating the exchanged data rely on sensors risk malfunctioning, potentially introducing wrong information into the network. 
Moreover, the intrusion of malicious actors into the driving system can also take place, potentially enabled through the wireless communication links.

One way of addressing these limitations is by recording the data exchanged among vehicles in order to analyze what faults in normal driving operations led to any incidents, such as accidents.
Legal authorities and insurance companies can then use such recorded data for accountability purposes. Original equipment manufacturers (OEMs) responsible for producing autonomous vehicles can also take advantage of these data to fix any hardware or software issues. Although solutions such as the Event Data Recorder (EDR) currently exist and are already deployed in modern cars, they do not yet capture V2X data \cite{edr}. Determining the specific data that should be recorded can also pose a challenge due to the substantial amounts expected to be generated \cite{cheng18}.

Additionally, vehicles should possess mechanisms enabling the establishment of agreements or consensus to address scenarios that require or benefit from the approval of multiple parties. 
Such scenarios can include selecting the best maneuver course of action, platoon leader election, or situational awareness validation such as reported event verification.
Reputation-based mechanisms \cite{ying22}\cite{suo20}, where individual vehicles are assigned reputation scores based on their perceived adherence to system rules, can enhance trust among road users and facilitate cooperative actions. However, relying solely on such mechanisms for cooperative systems may not be prudent, as vehicles with previously compliant behavior could potentially exploit their reputation scores for malicious intent.
Utilizing classical consensus algorithms, particularly Byzantine Fault Tolerance (BFT) ones, can be advantageous in reaching agreements in the presence of faulty nodes such as misbehaving vehicles. However, the complexity of such algorithms can pose challenges, especially in the context of local vehicular communications, where the communications channel may become overloaded or suffer imposing performance issues, especially in non-negligible packet loss scenarios.

Vehicular VAS necessitate a robust mechanism suitable for distributed systems, which can ensure the seamless occurrence of vehicle-to-vehicle (V2V) or vehicle-to-infrastructure (V2I) transactions. 
Ideally, transactions supporting vehicular services should take place locally in a peer-to-peer fashion, utilizing decentralized transactional systems that eliminate the necessity for third-party involvement.
Such services can include classical road services such as electronic toll collection (ETC) \cite{payasyougo}, emerging services such as electric vehicle (EV) charging \cite{shurrab22}, and even futuristic services such as rewards in cooperative maneuvers among fully connected autonomous vehicles (CAVs) \cite{sun21}. 
 
Solutions that are designed by taking into account these issues must be able to consider the opportunistic nature of vehicular networks and be as lightweight as possible. In this work, we explore a solution that is employed over the usual vehicular network traffic through the applicational extension concept. In the context of vehicular communications, we define an applicational extension as a payload originated in the applications layer of a communications protocol stack, issued and analyzed by a respective applicational protocol. The purpose of an applicational extension is to extend the features of the standardized protocols running on the lower layers of the protocol stack. Moreover, applicational extensions can introduce novel features by utilizing the existing transmitted packets as a transport medium, eliminating the necessity of transmitting complete separate packets. This approach significantly enhances the communication channel's efficiency and the ITS station's computational capability by decreasing the number of total message headers needed to be generated, transmitted, and analyzed.

We propose an applicational extension named Verifiable Event Extension (VEE), which can extend any type of ITS message. This extension seamlessly adds features associated with data security, consensus, and trading to the vehicular network without requiring any major overhaul of current communication protocols. To accomplish this, the extension data is divided into three modules that can be used separately or in any combination. For data security, a module named \textit{Ledger} is employed, which provides features associated with distributed ledger technologies (DLTs), such as data immutability and traceability. For consensus, a module named \textit{Consensus} is employed, allowing road users to reach non-critical agreements using classical consensus algorithms employed in distributed networks. For trading, a module named \textit{Token} is employed, which provides a basic cryptocurrency-based trading mechanism, enabling service exchanges among road users and possibly gamification features. 
This extension-based communications approach is based on and generalizes the local vehicular communications design of the VERCO (VERifiable COoperation) architecture \cite{verco}. This design offers a scalable and lightweight mechanism for data security in vehicular networks but relies on a specific maneuver coordination protocol to function. The approach presented here generalizes this mechanism to be applicable across various vehicular scenarios, while also introducing new consensus and trading features.

To address the mentioned issues, this document provides the following main contributions:
\begin{itemize}
    \item a communications protocol design based on the concept of applicational extensions to seamlessly add to or extend the features of standardized cooperative vehicular protocols with minimal overhead. We base our approach on the current ETSI ITS protocol stack;
    \item the Verifiable Event Extension (VEE), an applicational extension, and its modules, focused on enhancing security-related features of the vehicular network and ecosystem;
    \item envisioned protocols that employ VEE to complement cooperative vehicular protocols and scenarios;
    \item a performance-focused analysis of the proposed design using hardware-in-the-loop (HiL) simulations.
\end{itemize}

The rest of this document is organized as follows. Section \ref{sec:related-work} presents the state-of-the-art and an overview of related works. Section \ref{sec:system-model} provides the system model and assumptions of the proposed design. Section \ref{sec:vee} presents the design of VEE and an overview of its operation. Section \ref{sec:modules} provides the definitions of three modules in which features can be harvested to create advanced VEE-based protocols. Section \ref{sec:sps} presents three VEE-based protocols to enhance security-related features of typical cooperative driving scenarios. Section \ref{sec:simulations} provides a performance-based analysis of the proposed design using HiL setups. Section \ref{sec:conclusion} of this document presents the final conclusions drawn from the research findings and outlines the envisioned future work.
The list of acronyms used throughout this document is presented in Table \ref{tab:acronyms}.

\section{Related Work}\label{sec:related-work}

To the best of our knowledge, the concept of applicational extensions is new in vehicular communications. While message extensions are mentioned throughout V2X standards and allowed through the use of the ASN.1 extension marker ``\verb|...|'' in the ASN.1 definitions of V2X messages, their use remains mostly as a standardization feature for future message versions while also allowing for backward compatibility. 
The use of private extensions, i.e., non-standardized extensions, in such a way remains a discouraged feature \cite{facilities-protocols-standard}. Applicational extensions, however, expand the ITS message in a manner that preserves the original ASN.1 definition of the ITS message without using the extension marker. As a result, they do not introduce any potential complications to the integrity, interoperability, and analysis of the ITS message data.

\subsection{Extensions}

The use of other types of extensions, such as specific message extensions through modification of the message structure, is a popular topic in the field of vehicular communications. Message extensions, like Cooperative Awareness Messages (CAM) \cite{cam-standard} extensions, have been proposed to add features such as the future trajectory by Renzler et al. \cite{renzler20}. This work contributes to collision avoidance in platooning by extending the CAM message format to include the vehicle's future trajectory, and handle events beyond typical platoon operation. The authors introduce the CAM extension and compare its payload size with the standard, revealing a potential overhead of up to +90\%. They also outline event handling strategies. Simulations conducted with OMNeT++ and SUMO simulators validate both the CAM extension and event handling approaches. Although not analyzed in the paper, employing future trajectory-extended CAMs may still be a more practical approach, even when considering the +90\% overhead, particularly in terms of managing channel load, as opposed to also using a separate message such as Maneuver Coordination Messages (MCMs) \cite{mcs-standard}, which include said future trajectory. 
In \cite{hobert15}, Hobert et al. propose a CAM extension with data associated with several cooperative use-cases and ETSI ITS Facilities-layer protocols, like coordinated maneuvers and collective perception. CAMs can be here considered as a transport medium for other protocols' data, similarly to the design we propose in this document. Some of the proposed use-cases and protocols in the cited contribution have however progressed in recent years to their own standards employing dedicated ITS messages like MCMs and Collective Perception Messages (CPMs) \cite{cpm}. 
Similarly, Günther et al. \cite{gunther16} analyze the effectiveness of a CAM extension containing the collective perception data with respect to channel load and vehicles' awareness of the system elements. They also analyze this extension against a CPM-variant named Environmental Perception Message (EPM). Results show that the extended-CAM outperformed the EPM in most scenarios. 
Rondidone et al. \cite{rondinone18} propose a CAM extension for connected automated vehicles (CAVs) targetting platoon management. The system is integrated with the infrastructure, focusing on intersections support for traffic light signal timing optimization.
Loewen et al. \cite{loewen17} propose extensions to the CAM and DENM messages for added Vulnerable Road User (VRU) safety, namely focusing on cyclists. New station and event types are proposed through the use of sub-fields in both CAMs and DENMs, by employing the ASN.1 extension marker for backward compatibility reasons.

MCM extensions have also been a popular way to contribute to the topic of maneuver coordinations, being these a pillar of current standardization efforts. Correa et al. \cite{correa19} proposed an MCM extension which adds infrastructure support to coordinated maneuvers. Through this extension, roadside units (RSUs) can advise vehicles on the best course of action in any road segment which they oversee. Similarly, Mertens et al. \cite{mertens21} built on this concept, further extending the MCM, with a request-response mechanism for RSU-based maneuvers suggestions. Maksimovski et al. \cite{prima} also extend MCMs to enable a maneuver negotiation protocol. Using this protocol, vehicles can request maneuvers and respond to these requests. The requesting vehicle finalizes the negotiation by transmitting a confirmation message.

%In vehicular networks, the concept of piggybacking has been proposed to mainly share information across several networking nodes considering both the age of information and required packet hops. Such an example is work of Kaul et al. \cite{kaul11} which employes the concept of piggybacking to enable vehicles to efficiently share information by incorporating other vehicles' states, thereby optimizing information dissemination while considering data freshness and transmission hops. They demonstrated substantial reductions in system age and improved wireless link rates in multi-hop networks.

%The protocol which applies the applicational extension can also be viewed as a piggybacked protocol, which uses an underlying protocol for transport. This is different from the common concept of piggybacking in computer networks where, in the same protocol, outgoing acknowledgment packet is delayed and is instead included in the next outgoing data packet, reducing the number of transmitting packets by half. Piggybacking related with the use of two different protocols 
%Piggybacking? Tunneling?

Given the above overview of the current state-of-the-art on ITS message extensions, it is evident that there is no generalized applicable ITS message extension capable of facilitating the desired functionalities of data security, consensus and trading, or other application-specific features. The mentioned contributions require the modification of the established message format, thus in order to achieve interoperability a corresponding adjustment to the underlying standard is needed.
The system proposed in this study capitalizes on the existing lack of methodology for generalized extensions within vehicular networks by introducing a novel approach to address this challenge.

\subsection{Data security, consensus, and trading}

In the realm of data security, consensus protocols, and trading mechanisms within vehicular environments, numerous academic contributions have been introduced. The sheer volume of research in these domains within the context of V2X communications is substantial, rendering a comprehensive evaluation beyond the scope of this document. However, we outline selected contributions that bear greater relevance to our contribution.

For data security, the vast majority of contributions employ blockchain and distributed ledger technologies (DLTs) in some way or another \cite{mendiboure20}\cite{grover22}\cite{peng21}. In this work, data security is provided by the VERCO architecture \cite{verco}, which allows for vehicular data-security enhancements while being non-reliant on the infrastructure like RSUs, contrary to most blockchain-based contributions. It is based on a novel directed acyclic graph (DAG) DLT mechanism which records data related to cooperative maneuvers. This architecture is accompanied by a maneuver coordination protocol, wherein blockchain-related details are seamlessly integrated. The blockchain architecture however lacks generalization for the fact that is solely focused on cooperative maneuvers, relying on the proposed maneuver coordination protocol to work. Here, we generalize the blockchain integration into the vehicular communications through the proposed extension, which can be applied to both cooperative maneuvers and other types of vehicular applications.

Regarding consensus among vehicles, a substantial number of contributions provide non-BFT consensus algorithms mainly for coordinated maneuvers \cite{hafner22} due to their lower message complexity, required in safety-critical scenarios. Some BFT protocols however have been proposed. 
Similarly to the consensus use-case we provide in this document, Liu et al. \cite{liu19} provide a BFT-based algorithm for connected vehicles named BFCV in order to evaluate and verify reported vehicular events. The authors also propose a scheme named proof-of-eligibility in order to only allow vehicles which may have knowledge or view of the event to participate in the consensus process. 
In the same way, Cao et al. \cite{cao08} proposed another consensus-based scheme named proof-of-relevance which aims at proving which vehicles are relevant to an event and are allowed to report it.
%Li et al. \cite{li21} propose a BFT consensus algorithm named ABC (Approximation-based Binary Consensus) reach generalized agreements on a binary value, providing a limited analysis.
%For federated machine learning among vehicles, Chen et al. \cite{chen21} propose the use of a consensus protocol based on the distributed randomness-provider HydRand protocol \cite{hydrand}.
Sawade et al. \cite{sawade18} employ the BFT Turquois consensus algorithm \cite{turquois} for state synchronization among vehicles during coordinated maneuvers. The authors found that the proposed system is feasible up to 20\% packet loss, with a considerable performance degradation above these values.
Our past study \cite{mxm-pbft} employs PBFT \cite{pbft} among vehicles in order to record consensus-agreed cooperative maneuvers data in a DLT. 
Albeit a limited analysis is only provided, similarly to the results of the previous mentioned work, the performance of the consensus algorithm degrades rapidly above 20\% packet loss.

Regarding trading mechanisms in V2X, most approaches also employ DLTs and blockchains in one way or another \cite{jolfaei23} \cite{deng20} \cite{jabbar21}. Bartolomeu et al. \cite{payasyougo} have introduced an ETC mechanism that leverages 5G device-to-device (D2D) communications and relies on the IOTA DLT \cite{tangle}, specifically tailored for the Internet of Things (IoT) domain. Notably, it aims at executing the tolling transaction locally. This approach involves the end user's smartphone, carried within the vehicle, conducting the financial transaction locally in conjunction with the RSU. This transaction in the IOTA framework employs a proof-of-work consensus mechanism, computed by the smartphone. One use-case that we provide in this work is similar to this approach, though it employs an on-board unit (OBU) present in the vehicle, being responsible for the provision of vehicular communications. Current V2X tolling mechanisms \cite{sae-j3217} are also employed in conjunction with lighter cryptographic mechanisms, without the need for performing the computationally intensive proof-of-work.
%Leiding et al. \cite{leiding18} propose a blockchain based trading mechanism for V2X applications. Among vehicles, they mention that such system could be employed to reserve road space, negotiating and trading road segments as a good. 

Overall, aspects like channel performance, implementation and testing in common OBU/RSU hardware, and overall generalization of the proposed contributions, are factors which are mostly ignored regarding contributions related with data security, consensus, and trading in V2X. We try to tackle these issues by providing the design of generalized modules which provide said features, whilst also validating their feasibility using hardware-in-the-loop simulations.

\section{System Model and Assumptions}\label{sec:system-model}

In this section, we present the system model for the proposed design and outline the key assumptions considered in our study. While our design is based on the current ETSI ITS protocol stack \cite{etsi-its}, we also provide some details of its implementation implications regarding the IEEE WAVE protocol stack \cite{wave}.

\subsection{Assumptions}

The concept of applicational extensions is not envisioned in the current ITS standards, however current ITS message decoder implementations are, in principle, able to safely decode a standard ITS message with any extra appended data (such as the VEE) since these messages have a predefined structure with constant length data fields or variable length fields with their length declared in the message itself. Any extra appended data (the extension) not consumed by the ITS message decoder can be discarded or further analyzed if the ITS-S employs the associated applicational protocol. We therefore assume that this proposal of applicational extensions do not invalidate the correct decoding of extended ITS messages and the features provided by the existing ITS protocols. 

The proposed design is independent of the radio communications access technology in vehicular networks, with current mainstream technologies being ITS-G5 and C-V2X. The applicational extension data should be small enough (in length) so that it does not considerably impact the average end-to-end communication latency between stations. This can even be assumed taking into account a couple of considerations. Let $L_{E}$ be the length of the extension and $L_{P}$ the total length of the packet without the extension. The channel load overhead $O_+$ (\%) of the extension per packet is,

\begin{equation}
    O_+ = \frac{L_E}{L_P} * 100
\end{equation}

Trivially, the smaller is $L_E$ the smaller is the overhead. If we assume that extended packets are only sent occasionally, i.e., road users exchange both non-extended and extended packets, the expected value of the overhead is, 

\begin{equation}
   E(O_+) = \frac{L_E N_{P+}}{L_P N} * 100
\end{equation}
where $N_{P+}$ is the number of extended packets, and $N$ the total number of packets.

Considering $L_E \ll L_P$, then the impact of the extension is minimal. Considering also that $ N_{P+} \ll N$ then the overall impact of the extension is even less. For example, consider one station transmitting CAMs, with an average length of 400 bytes, at a rate of 10 Hz. Consider also that only one of these CAMs is extended per second (1 Hz, 10\%), with an extension payload of 50 bytes. The average overhead will be only $+1.25\%$. At the default ITS-G5 channel bitrate of 6 Mbps, this equates to an extra 67 $\mu s$ transmission time, per second. 
Naturally, such assumptions may not always hold true, depending on the analyzed time frame and current vehicular scenario.

Nevertheless, we evaluate the impact of our design using a performance analysis.
A performance analysis is especially important in vehicular networks due to the safety-critical nature of the involved communications:
\begin{enumerate}
    \item Safety: potential issues should be resolved as fast as possible;
    \item System state evolution: the  employed communication protocols should finish as fast as possible, considering the dynamic nature of the vehicular environment, as prolonged protocol execution can potentially compromise its intended functionality;
    \item Channel efficiency: the capacity of the wireless communication channel is limited, and it can only accommodate a certain amount of data being transmitted.
\end{enumerate}

A level of security is assumed. The standard-compliant security entity of the protocol stack should cryptographically sign the exchanged messages in the vehicular network. In both ETSI ITS and IEEE WAVE protocol stacks, the ECDSA signing algorithm provides authentication and asserts a station's permissions in the vehicular network. If sensitive or private information is required to be exchanged, exchanged messages must also be encrypted. In both ETSI ITS and IEEE WAVE protocol stacks, the ECIES encryption algorithm provides an encryption feature and establishes private communication channels. Both the exchanged ITS message data and applicational extension should be able to take advantage of both of these security features, where it is assumed that every message is signed/authenticated. Revoked or misbehaving stations should be correctly handled by the security entity, and originating messages from these should not reach the upper layers of the protocol stack. Nevertheless, the proposed design can help identify misbehaving stations even if the security entity fails to do so.

Regarding the level of vehicular automation, it only depends on the level of automation required by the extended standardized ITS protocols themselves. It is not envisioned that the applicational protocols need or can affect the control mechanisms of the vehicle, at least in a significant, direct, or critical manner. 

\subsection{Requirements}

The proposed design requires that the entity generating ITS messages must first request to the applicational extension protocol any extension to be added to the ITS message before forwarding it to the lower layers of the protocol stack. In the ETSI ITS protocol stack, this means that the Facilities \cite{facilities-standard} layer/service generating messages such as CAMs \cite{cam-standard}, DENMs \cite{denm-standard}, etc., must request to the Applications layer/service the extension before forwarding the message to the lower Network and Transport layer/service. In the IEEE WAVE protocol stack, both the protocol generating ITS messages (e.g., BSMs \cite{bsm-standard}) and the applicational extension protocol coexist in the Applications layer, and the former should request the extension to the latter before forwarding the message to the Network and Transport layer.

Eventually, the extension-related operations could be shifted from the applicational layer to the lower layers of the protocol stack. However, this would require more widespread implementation of extension-related logic, backed up by standardization efforts, with the possible added constraints on the liberty of adding, modifying and implementing new types of extensions with a more \textit{applicational} purpose.

\section{The Verifiable Event Extension (VEE)}\label{sec:vee}

VEE is an applicational extension that appends to standard ITS messages. It is a complementary component to ITS messages and is managed by the Verifiable Event Protocol (VEP). VEP is responsible for managing, extending, and analyzing ITS messages incorporating the VEE. 
%While VEP can extend any type of vehicular message, running on whichever protocol stack, here we focus the ETSI protocol stack and related protocols.

In the ETSI ITS protocol stack, VEP runs in the Applications layer. A representation of VEP in the protocol stack is presented in Figure \ref{fig:stack}. The design and implementation of VEP and VEE aim to be as standard compliant as possible so as not to affect the normal functioning of ITS-Ss not employing the VEP protocol. VEP extends ITS messages with the VEE, providing additional features to the vehicular network while maintaining the data integrity of these messages. A representation of an ETSI ITS message packet data unit (PDU) extended with the VEE is presented in Figure \ref{fig:pdu}.

\begin{figure}
    \centering
    \includegraphics[width=0.9\linewidth]{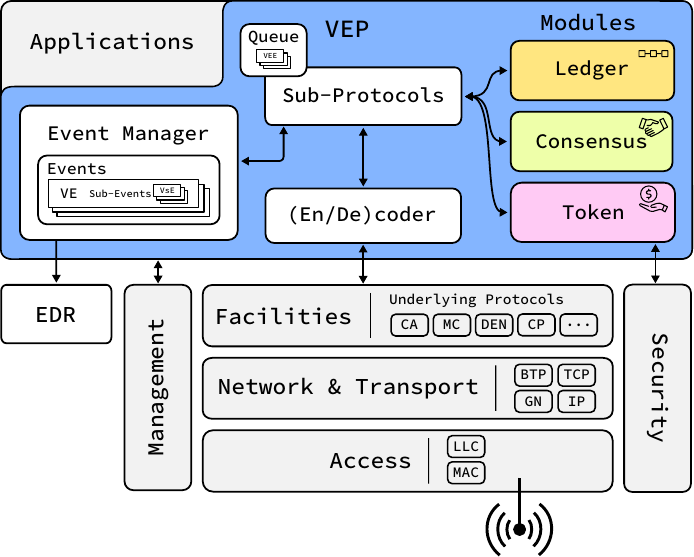}
    \caption{ETSI ITS protocol stack with the Verifiable Event Protocol (VEP).}
    \label{fig:stack}
\end{figure}

\begin{figure*}
    \centering
    \includegraphics[width=\textwidth]{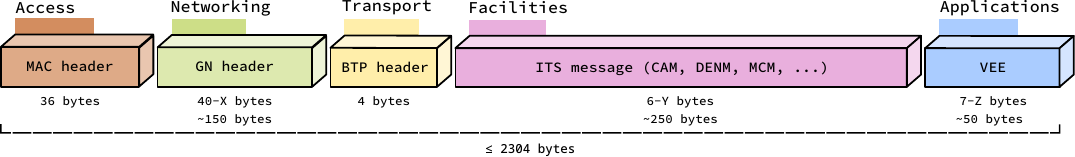}
    \caption{PDU of an ETSI ITS message extended with VEE. The header sizes of each PDU component are shown. For the GeoNetworking (GN) header, ITS messages, and VEE, only the minimum sizes are shown since these have variable sizes. For these, the \textit{typical} sizes are also shown, however these values can vary a lot according to several factors such as security parameters (GN), message type (ITS message), and current vehicular event type (VEE).}
    \label{fig:pdu}
\end{figure*}

VEP keeps track of vehicular events using the concept of Verifiable Events (VEs) generated following specifications provided by VEP sub-protocols (SPs). Each VE is associated with a single type of vehicular event, such as a cooperative maneuver, identified by an event ID. It can be decomposed into several (verifiable) sub-events (VsEs), associated with several instances of a vehicular event. For example, a cooperative maneuver event can be decomposed into several sub-events associated with, the maneuver coordination or negotiation, the maneuver execution, and the maneuver confirmation.

An SP generates VEs and VsEs by analyzing the current vehicular context, primarily by examining the received ITS messages. It is also responsible for issuing VEEs and choosing which ITS messages will be extended with these. An SP is defined based on the objectives of a particular event type. Its design and implementation are fundamentally constrained by the underlying protocols (UPs) responsible for generating the extensible ITS messages. 

SPs issue VEEs by utilizing functions of independent entities named modules. These modules provide the main features of VEP and supply data included in each VEE.
The structure of VEE is represented in Figure \ref{fig:asn1}.
The structure of VEE comprises an event ID (integer), an SP ID (integer), and various containers (here, the Ledger, Consensus, and Token), each associated with a module named equally to the container it supplies. These containers can be utilized independently or in combination to handle different types of vehicular events. The selection of specific modules depends on the specification of the SP associated with a particular vehicular event. 

 \begin{figure}
    \centering
    \includegraphics[width=\linewidth]{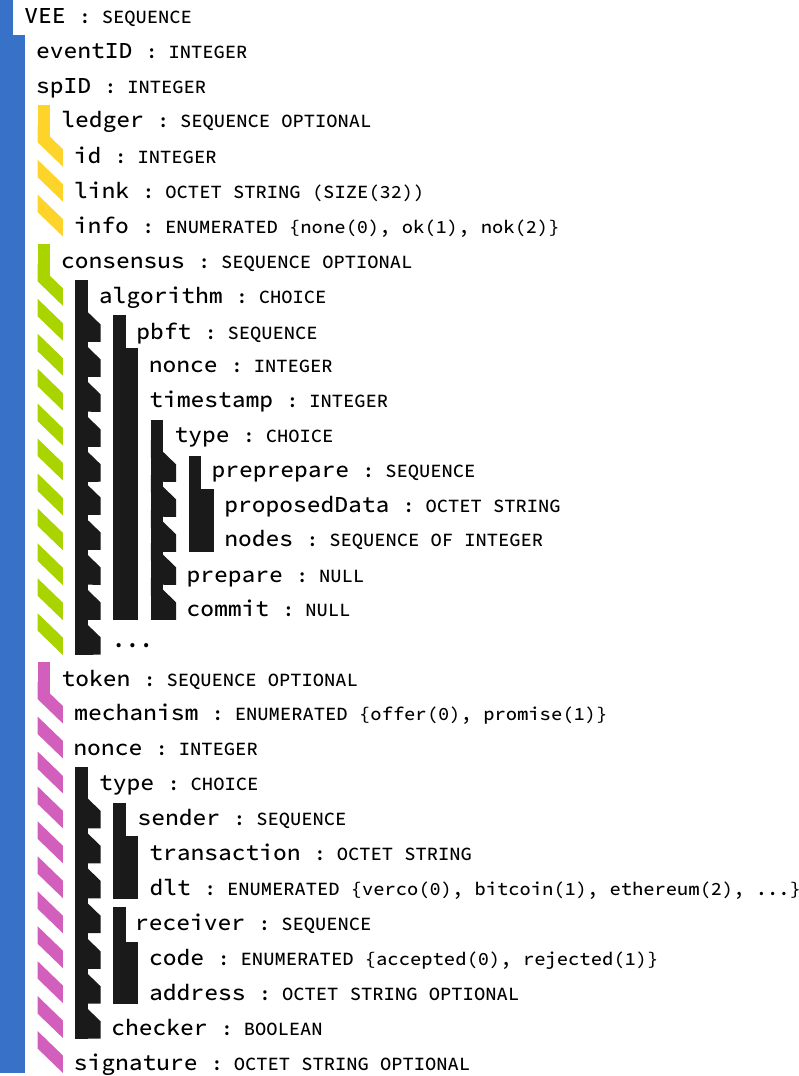}
    \caption{A representation of the VEE ASN.1 message format. A full bar represents a \texttt{SEQUENCE} field, a dented bar represents a \texttt{CHOICE} field, and a dashed bar represents an \texttt{OPTIONAL} field.}
    \label{fig:asn1}
\end{figure}

It is crucial to emphasize that the functioning of an SP must remain independent of the UPs, with any potential impact on UPs being minimal and negligible in terms of performance. The failure or malfunction of an SP process must not result in the invalidation of the operation or the integrity of the UPs. In particular, the total data size added by the extension to the ITS message must also mind and comply with the maximum packet sizes limited by wireless access technology and lower layer protocols.

Ultimately, VEP, responsible for handling V2X messages and events, can be integrated to support the EDR module of a vehicle. By selectively processing vehicular events, VEP can identify and record potentially important VEs in the EDR. This integration enhances the EDR's existing physical dynamics-based recorded features \cite{edr-standard}, providing a more comprehensive and verifiable dataset for post-event analysis and investigation.

Two different operational modes, named \textit{interactive} and \textit{passive} modes, are associated with the use of VEE, whose utilization depends on the ITS message type and its data. A representation of ITS message handling by the SPs and subsequent VE generation is presented in Figure \ref{fig:msg-path}. Each mode is detailed below.

\begin{figure}
    \centering
    \includegraphics[width=\linewidth]{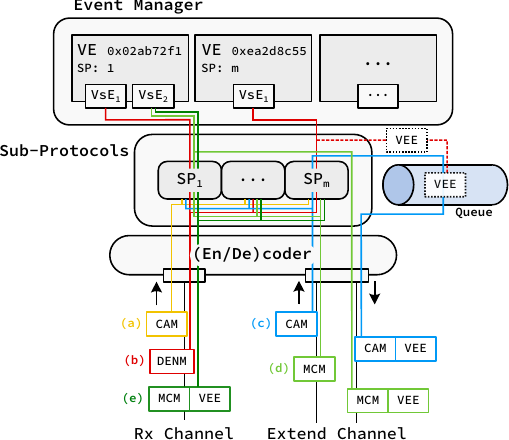}
    \caption{Representation of the message handling process in VEP. Possible paths of ITS messages as the protocol handles them are shown. (a) Not all received messages are deemed important by the protocol. (b) Some messages are handled by multiple SPs. In this case this DENM message triggers SP$_m$ to create a passive operation process, adding a VEE to the queue, while SP$_1$ simply saves it. (c) A queued VEE is added to a message that could potentially be extended. (d) Some messages are extended using the interactive operation mode, spending no time in the queue. (e) VsE data can be composed of several messages, extended or not.}
    \label{fig:msg-path}
\end{figure}

\subsection{Interactive mode}
VEP analyses the ITS message contents and appends an associated VEE. This VEE ultimately extends the UP with related data. Interactive VEEs are issued with spontaneous, event-triggered ITS messages.
Only one SP can extend each ITS message, being this condition also valid for the passive mode.
    
\subsection{Passive mode}\label{sec:passive-mode}
VEP fetches any internal queued data and creates a VEE, appending the ITS message with it. The data in this VEE can be associated with any applicational protocol that is unrelated to the lower-layer underlying  protocol. 
Given the lack of association between this data and the underlying V2X safety-critical protocols, it's arguable that transmitting such data could be performed using alternative channels, like the Service Channel (SCH), instead of the Control Channel (CCH) where most standardized ITS messages are transmitted. Nevertheless, this operation mode can still have a positive impact on overall channel load management on multi-channel systems, be it employed over CCH or SCH -transmitted messages.
This mechanism is used mostly to increase channel and spectrum efficiency through reduction of the number of broadcasted packets. This means that a new packet with all the  required protocol stack headers does not need to be issued just for the applicational protocol logic. 
Passive VEEs are issued with periodic, time-triggered ITS messages. Eventually, they can also be used over event-triggered ITS messages for which the SPs do not issue interactive mode VEEs. Our analysis below on passive VEEs only considers periodic ITS messages, as these constitute the vast majority of V2X traffic in most vehicular scenarios. %Moreover, the unpredictability of event-triggered ITS messages

The performance of VEE-associated applicational SP is associated with the performance of the UPs used as a transport medium. Let us assume a combination of UPs with variable transmission message period $T$, with expected value $E(T)$, with a probability density $f(t)$, cumulative distribution $F_T(t)$, and bounded in $\left[T_{min}, T_{max}\right]$.

Let us assume \circled{1} a packet processing time $\tau_{ps}$ and a packet transmission time $\tau_{tx}$ much smaller than the possible values of $T$, $\tau_{ps} + \tau_{tx} \ll T$.

A queued VEE $q_j$ at the position $j = 0$ in the queue $q$ ($j \in \{0,\ldots,\texttt{len}(q)-1\}$) will be transmitted with a delay $\tau_{q_0}$ bounded in $\left[0, T_{max}\right]$. 
That is, 0 if VEE issuance coincides with UP message issuance, and $T_{max}$ if VEE issuance needs to wait the maximum period of the UP. The traffic generated by the UP can also be viewed as a renewal process, where $\tau_{q_0} = W$ can be considered the waiting time or residual time of the process. Following \cite[Chapter~3]{gallager12}, the expected value of $\tau_{q_0}$ is, 

\begin{equation}
    E(W) = \frac{E(T^2)}{2E(T)}
\end{equation}
with a cumulative distribution function,
\begin{equation}
    F_W(w) = \frac{1}{E(T)} \int^w_0 (1-F_T(t)) dt
\end{equation}
If $T$ is a discrete variable then $F_W$ can be calculated as sum over the possible $t$ values,
\begin{equation}
    F_W(w) = \frac{1}{E(T)} \sum_{t_i}^w \left[ (1-F_T(t_i))  (t_i - t_{i-1}) \right]
\end{equation}

Subsequent queued VEEs ($j > 0$) will be transmitted with an average delay,
\begin{equation}\label{eq:queued-vee}
E(\tau_{q_j})= E(W) + jE(T)
\end{equation}

A representation of the UP traffic and the delays of the queued VEEs are presented in Figure \ref{fig:up-vee}.

\begin{figure}
    \centering
    \includegraphics[width=1\linewidth]{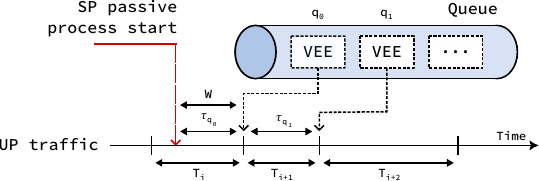}
    \caption{UP traffic and VEE queue delay representation.}
    \label{fig:up-vee}
\end{figure}

The overall process latency of an SP implementing a multi-stage process with multiple message exchange rounds using the passive mode is similar to that of the above formulas.
A communications protocol can be divided into several stages, where each stage can be further divided into several message exchange rounds. 
A message exchange round can itself contain non-concurrent messages, e.g., a request (or various concurrent requests) and a response (or various concurrent responses) or a single (or various concurrent) message(s). 
Let us call a concurrent message exchange a \textit{sub-round}. 
Assuming \circled{1}, and that there \circled{2} is no packet loss/message retransmission, the overall delay of a process with $k$ sub-rounds is, 

\begin{equation}\label{eq:passive-delay}
\tau_p = \sum_{i=1}^{i=k} \tau_i
\end{equation}
where $\tau_{i}$ is the delay of the sub-round $i$. For a rough estimation of the average $E(\tau_p)$, $E(\tau_i)=E(T)$ can be used, giving us $E(\tau_p)=E(T)k$.  $\tau_p$ will also be bounded $\left[ T_{min} (k-1) , T_{max} k \right]$. However, a process to advance to the next stage or round usually requires that $m$ total sub-round messages or $g\leq m$ valid messages be received, such as in the case of fault tolerance algorithms. This means that $\tau_i$ will be as fast as it takes for $g$ valid messages to be exchanged. Specifically, it will equal the delay of the $g$th message. Generalizing both $W$ and $T$ distributions with the variable $Y \in \{T, W\}$, the delay distribution of $g$th message can be described through order statistics. A order statistic $Y^m_{(g)}$ is the $g$th largest $Y$ value in the ordered sample $\{Y_1, Y_2, \ldots, Y_{m}\}$. 
%The delay $\tau_i$ then depends on the amount of messages $g$ needed to be exchanged and the probability distribution of the underlying protocol period $T$ (or the respective residual time $W$). 
Given $F_Y(t)$, the cumulative distribution function of $Y^m_{(g)}$ \cite[Chapter~2]{order-statistics} is,

\begin{equation}
    F_{Y^m_{(g)}}(t) = \sum^m_{j=g} \binom{m}{j} \left[ Y_T(t) \right]^j \left[1- F_Y(t) \right]^{m-j}
\end{equation}

If all messages need to be exchanged in the sub-round, then $g=m$, and the above is simplified to,
\begin{equation}
    F_{Y^m_{(m)}}(t) = \left[ F_Y(t) \right]^m
\end{equation}

Assuming $Y$ is a discrete variable, the expected value of $Y^m_{(g)}$ is the sum over of the probabilities of the possible $t$ values, 
\begin{equation}
    E(Y^m_{(g)}) = \sum_{t_i} \left[ (F_{Y^m_{(g)}}(t_i) - F_{Y^m_{(g)}}(t_{i-1}))*t_i\right] 
\end{equation}
%where $F_{T^m_{(g)}}(t)$ is the cumulative distribution of $t$ and $F_{Y^m_{(g)}}(t-1)$ the cumulative distribution at the previous value of $t$. 

The average total process delay $\tau_p$ is then,
\begin{equation}\label{eq:delay-ksrounds}
     E(\tau_p) = \sum_{i=1}^{i=k} E(Y^{m_i}_{(g_i)})
\end{equation}
where $Y$ is either $W$ or $T$, $g_i$ is the number of messages required to be exchanged and $m_i$ the total possible number of messages, at sub-round $i$. While $Y=W$ can be chosen in empty queue scenarios, i.e., low traffic and beginning of the passive mode process, $Y=T$ can be used to estimate delays in higher traffic scenarios or in subsequent sub-rounds, possibly taking also into account the $j$ length of the queue (adding $\tau_{q_j}$ to $\tau_{p}$). This choice will heavily rely on the logic of the SP employing the passive mode. Mixed scenarios where several stations have different queue lengths are out of the scope of this analysis.

\section{Envisioned Modules}\label{sec:modules}

VEP Modules provide the main features of the proposed design. Each module is an independent entity (see Figure \ref{fig:stack}), which SPs can use separately or simultaneously for a multitude of use cases. Each module is also responsible for analyzing and generating the VEE containers' data. In this work, we envision three different modules: 
the Ledger module, to increase vehicular data security; the Consensus module, to facilitate consensus-building mechanisms that allow participants to reach agreements on shared data, rules, or actions; and the Token module, to provide a framework for the secure and efficient exchange of digital assets or services.

\subsection{Ledger}\label{sec:ledger-module}

The Ledger module provides information related to a distributed ledger in order to enhance data security. 
Data security in vehicular networks is largely employed using blockchains and distributed ledger technologies (DLTs) \cite{mendiboure20}\cite{grover22}\cite{peng21}. In this work, we establish data security through the utilization of the VERCO architecture \cite{verco}. 
Notably, VERCO offers enhancements in vehicular data security without a dependency on infrastructure elements such as RSUs. The core of VERCO's data security framework lies in a novel DLT mechanism referred to as \textit{localchain}, which serves as a repository for cooperative maneuver-related data. This mechanism empowers vehicles, RSUs, and other entities within the local road environment to record critical vehicular data with replication properties in a lightweight, immutable, and traceable database. Ultimately, it serves as a tool to resolve and pinpoint instances of missing data as vehicular data is linked to each other. It also provides universal identifiers of vehicular events as several entities produce the same identifier (block hash) for the same event in a distributed manner.

A localchain is a geographically-based distributed ledger that adopts a directed acyclic graph (DAG) structure to address the high scalability demands of vehicular networks. It utilises sharding techniques \cite{sharding}, whereby each station stores only a portion of the data of the complete localchain. This technique is necessary because a station cannot collect and store all the vehicular data exchanged due to range restrictions, packet loss scenarios, or storage limitations.
VEE generalizes the maneuver-focused communication protocol and associated messages proposed in the original contribution of localchains, i.e., instead of Verifiable Cooperation Messages (VCMs) encompassing ledger information, VEE now encompasses this information instead, enabling the use of localchains over other maneuver coordination protocols, or even other types of vehicular communication protocols.

The Ledger module is adapted to the localchain technology. It keeps track of the various system localchains and groups V2X messages into localchain blocks. Each block is associated with a VsE and is composed of the ITS messages that are associated with that sub-event.
The SP must explicitly define the types of and which specific messages, as well as their order in each block. 
Every ITS/block message is mandatorily signed by the standard, thus validating and verifying the authenticity and provenance (who issued the message) of the data recorded in each block.

A block can only be created if the station received every message specified by the SP for that VsE. Otherwise, if a block is created using different messages, or a different number of messages, or even following a different message order, the station will create a different block from its peers, resulting in a different block hash, ultimately adding inconsistency to the localchain.

In the local communications among vehicles, the VEE's Ledger container provides contextual information about where to link data (which block to link the new localchain block to) and some limited information about the linked block. It comprises the localchain ID, the previous block's hash (SHA-256) to link the new one to, and an informational flag related to the previous block, providing additional context depending on the used SP.

A representation of a localchain composed of blocks containing VE data is presented in Figure \ref{fig:localchain}.

\begin{figure}[b]
    \centering
    \includegraphics[width=1\linewidth]{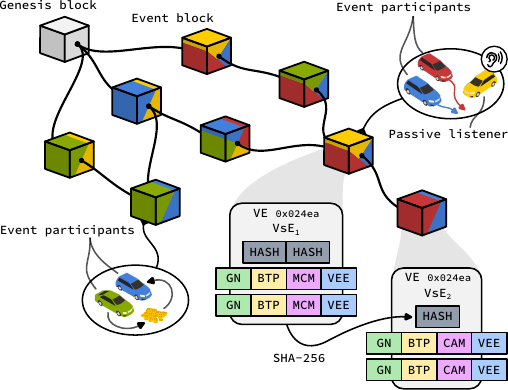}
    \caption{Representation of a localchain and its sharding feature. Different block colors represent vehicles which have the full block data, i.e., the ones which received all the required packets to produce the block.}
    \label{fig:localchain}
\end{figure}

\subsection{Consensus}\label{sec:consensus-module}

The Consensus module handles consensus processes, providing an agreement mechanism for various applications.

Given the possibility of packet loss, a fault-tolerant consensus algorithm can be recommended. Byzantine fault-tolerance (BFT) algorithms can also provide enhanced resilience against the presence of malicious nodes. 
BFT algorithms are however associated with protocols with larger message complexity to be considered for safety-critical scenarios. These protocols typically require an higher number of messages to be exchanged, when compared to lower complexity communications protocols employing simpler mechanisms such as request-response. 
Given the complexity of such algorithms and their probable disparity with the simplicity of typical vehicular communication protocols, it is envisioned that most SPs employing the Consensus module will operate using the passive operation mode. Therefore, the performance of a consensus process will be constrained by the amount of traffic available generated by the UPs. 
Nevertheless, we here consider one BFT algorithm, in order to analyze the performance of such algorithms running on the proposed extension-based communication system for non-safety-critical applications. 

We take as an example one of the most popular BFT consensus algorithms, the Pratical Byzantine Fault Tolerance (PBFT) \cite{pbft} protocol. This protocol is divided into three stages, the pre-prepare, the prepare, and the commit, and employs retransmission logic to account for any missed packets. It is most likely implemented using the passive mode for operation since its complexity and logic may be too different from any vehicular UP to be implemented using the interactive operation mode. In this protocol, with $n$ nodes, the primary node issues one pre-prepare message (in the pre-prepare stage) with a proposal.  Then, the $n-1$ replica nodes broadcast a prepare message (in the prepare stage) as an acknowledgment and in agreement with the proposal. Then, in the last stage (the commit) all $n$ nodes issue a commit message, confirming the receipt of the prepare messages of the other nodes. Following the nomenclature of Section \ref{sec:passive-mode}, the PBFT process has $k=3$ sub-rounds. PBFT is also fault-tolerant, meaning it can still advance if it does not receive the $m$ total messages of the sub-round. Specifically, a node advances in the process if it receives at least $2f$ messages per stage, where $f=\frac{n-1}{3}$ is the number of faulty nodes. 
A faulty node is a node that, either intentionally or unintentionally, deviates from correct behavior. A faulty node may delay, drop messages or send incorrect messages.
Given that we have 3 sub-rounds, the overall process delay will be $\tau_p=\tau_1 + \tau_2 + \tau_3$, assuming \circled{1}, and \circled{2}. If we also assume \circled{3} empty queues during the PBFT process and a \circled{4} single process running throughout its duration (i.e., no parallel consensus processes running by the same stations), we can approximate the average $\tau_p$ using Eq. \ref{eq:delay-ksrounds} as, 
\begin{equation}\label{eq:pbft-delay}
    E(\tau_p) = E(W^1_{(1)}) + E(W^{n-1}_{(2f)}) + E(T^{n}_{(2f+1)})
\end{equation}
being $E(W^1_{(1)})$ equal to $E(W)$. From the point of view of a node participating in the PBFT process, the overall delay may be slightly lower than the above approximation. For example, from the point of view of the primary node (or a replica), $E(\tau_3)$ is better approximated as $E(T^{n-1}_{(2f)})$ since it only needs to wait for $2f$ messages of the other $n-1$ nodes in the commit stage (it instantly ``receives'' its own commit message).

If a node does not receive the number of expected messages in a stage, it can retransmit the message associated with that stage after some delay $\tau_d$. 
In this case, the total average delay of a PBFT process is, assuming \circled{1}, \circled{3}, and \circled{4},
\begin{equation}\label{eq:pbft-delay-rt}
\tau_{pr}=  \tau_p + r\tau_d
\end{equation}
where $r$ is the number of repeated rounds/stages.

The Consensus module running on top of the VEE's passive operation mode should, in principle, support any type of consensus algorithm. Indeed, consensus algorithms designed for partially synchronous networks, characterized by periods of both synchrony and asynchrony, where message delays and node failures can occur but are eventually limited, are expected to retain their essential properties of safety (e.g., agreement and validity) and liveness (e.g., termination and progress) when operating on top of a time-bounded UP.

\newtheorem{theorem}{Theorem}
\begin{theorem}[Safety preservation]
Any consensus algorithm designed for partially synchronous networks running on top of a time-bounded UP preserves its safety properties.
\end{theorem}
\begin{proof}

We consider a consensus algorithm A designed to operate in a partially synchronous network, where message delays and node failures can occur but are eventually limited. Algorithm A runs on top of an underlying time-bounded UP protocol that facilitates the communication of consensus-related messages in a time-bounded manner.

For the sake of contradiction, let's assume that algorithm A when running on the time-bounded UP, violates its safety property. This means there exists an execution in which two correct nodes decide on different values $v_1$ and $v_2$ for the same consensus instance.

During this execution, since UP is time-bounded, message delays and node failures are eventually limited. As a result, nodes experience periods of synchrony where message 
delivery is guaranteed and delays are bounded.

In these synchronous periods, algorithm A behaves similarly to when it runs in a completely synchronous system. Thus, all correct nodes must eventually agree on the same value for the consensus instance during synchronous periods.

Since the algorithm A operates correctly during synchronous periods, it follows that any disagreement in decisions during asynchronous periods must be caused by message delays or other asynchrony-related issues.

However, because UP is time-bounded, these asynchronous periods cannot persist indefinitely, except if a node leaves the network causing the consensus algorithm to restart due to timeouts. 
Hence, asynchrony is eventually bounded and resolved, allowing algorithm A to achieve agreement.

Therefore, algorithm A preserves its safety property when running on the time-bounded UP.
\end{proof}

\begin{theorem}[Liveness preservation]
Any consensus algorithm designed for partially synchronous networks, running on top of a time-bounded UP, preserves its liveness properties.
\end{theorem}
\begin{proof}
To demonstrate liveness preservation, we assume, for the sake of contradiction, that algorithm A violates its liveness property when executed on the time-bounded UP. This means there exists an execution in which no correct node ever decides on a value for a valid consensus instance despite all correct nodes proposing values.

Again, due to the time-bounded nature of the underlying UP protocol, asynchronous periods are eventually limited and resolved. During synchronous periods, algorithm A behaves as if running in a completely synchronous system, where it guarantees progress and termination.

Since synchronous behavior is guaranteed by the time-bounded UP, we can conclude that algorithm A will eventually decide on a value during synchronous periods if all correct nodes propose values.

Thus, any violation of liveness during asynchronous periods must result from message delays, such as larger VEE queues and the temporary presence of asynchrony, which will eventually be resolved.

Therefore, algorithm A preserves its liveness property when running on the time-bounded UP.
\end{proof}

%Consensus algorithms designed for partially synchronous networks, characterized by periods of both synchrony and asynchrony, where message delays and node failures can occur but are eventually limited, are expected to retain their essential properties of safety (e.g., agreement and validity) and liveness (e.g., termination and progress) when operating on top of a time-bounded UP.

\subsection{Token}\label{sec:token-module}

The Token module handles the token trading features of the system, such as cryptocurrency trading, enabling the exchange of services among stations. Such services can include for instance Electronic Tolling Collection (ETC) and benefits in cooperative maneuvers through negotiation. While primarily designed for the token of the VERCO architecture, this module can also be easily adapted to work with other types of cryptocurrency. 

In the VEP implementation, we assume that the token-related identity is not bound to the vehicular identity, as it can raise privacy concerns. The Token module must have an interface to the Security layer/service to request the signing of token-related data. The Security layer/service must then be pre-loaded with the required cryptographic keys associated with the possible DLTs to be used. 

The Token module supports two different transaction mechanisms, the \textit{offer} and the \textit{promise}. 

\subsubsection{Offer}
The offer mechanism is similar to common cryptocurrency transactions, where the holder of the tokens creates a single transaction comprised mainly of the number of tokens to transfer, the output address(es) to transfer the tokens to, and its own signature, allowing the transfer. This mechanism easily supports popular DLTs such as Bitcoin \cite{bitcoin} or Ethereum \cite{ethereum}, assuming that the transaction complexity (e.g., number of outputs), is kept low enough so that the transaction size can fit in the packet to be transmitted using the access technology (ITS-G5, C-V2X, etc.). This is feasible, assuming a few hundred bytes per transaction.

\subsubsection{Promise}
The promise mechanism provides a way to only transfer tokens if an event is successful. The transaction is divided into two stages, starting with the token proposal and finishing with the verification of the event. In the first stage, the token holder issues the first part of the transaction containing the offer, and optionally the output addresses, and the signature. If the holder chooses to omit the output addresses, then the to-be token receivers can issue VEEs containing the output addresses, signed, to be appended to the first part of the transaction. If this is the case, the SP must correctly identify who participates in the event to avoid any nearby vehicles not relevant to the event issuing their own output address and receiving unwarranted tokens. This scheme is advantageous since the token holder may not be aware of the addresses of the recipients, and it also provides the freedom for the recipients to choose which address to receive tokens on. In the second stage, the event is complete, and the participants issue VEEs verifying if it finished successfully or not, creating the second part of the transaction allowing for the token transfer if the majority agrees that the event indeed finished successfully. The promise mechanism requires that the used DLT is programmed to understand this double-stage transaction, such as through the use of smart contracts, in order to allow the correct transfer of tokens.

\section{Envisioned Sub-Protocols (SPs)}\label{sec:sps}

This subsection illustrates practical applications of VEE that enhance or add various features to the vehicular network. We present three specific SPs, each centered around a previously described module, and discuss the possibility of integrating multiple modules. The three SPs covered in this section are the Maneuver SP, the View SP, and the Tolling SP.
The Maneuver SP demonstrates how VEE, primarily utilizing the Ledger module, can bolster cooperative maneuvers' data security.
The View SP explains how VEE, utilizing the Consensus module, can be used to establish a distributed approach for vehicles to improve their system-wide visibility or verify the state of the system.
Lastly, the Tolling SP showcases how VEE, utilizing the Token module, can facilitate the instant execution of road toll payments.

\subsection{Maneuver SP}\label{sec:maneuver-sp}

ETSI ITS Working Group 1 (WG1) is currently working on standardizing a vehicular communications protocol for coordinated maneuvers named Maneuver Coordination Service (MCS) \cite{mcs-standard}. It is based on a periodic sharing of the planned future trajectory by the vehicles through messages named Maneuver Coordination Messages (MCMs). Using this protocol, vehicles can also request certain maneuvers (by defining the desired trajectory) through a request-response mechanism. In this protocol, three different types of vehicles can be identified: the subject vehicle (SV), which requests a maneuver using a MCM.request; the target vehicle(s) (TV) participating in the maneuver, which either accept or reject the request with MCM.response(s); and the relevant vehicles (RVs) which drive nearby without participating in the maneuver.

Coordinated maneuvers, however can fail, and related data should be recorded in order to find out what went wrong and who (or what) should be held accountable.
Here, we propose an SP for coordinated maneuvers to increase related data security. This SP employs the Ledger module, with optional Token and Consensus modules. VEs of this SP are composed of at least two VsEs: the maneuver coordination, the maneuver verification, and optionally more VsEs with extra maneuver data.

During maneuver coordination (a.k.a., agreement seeking or maneuver negotiation), the SV issues a MCM.request extended with the VEE with the Ledger module. The Ledger module includes the ledger ID associated with the localchain of that geographic region and a previous block hash which the vehicle has data. Upon receiving the MCM.request which contains the station IDs of all TVs, these reply with a MCM.response with a VEE with the Ledger module, with the same localchain ID as the request and a previous block hash equal or different to the request. Upon receiving all messages (MCM.request and MCM.responses), maneuver participants and nearby stations create a block with all the exchanged MCM.request/response(s) messages in this first VsE. The order of the messages in the block is ascending, according to the station IDs of the maneuver participants. A user can only create a block if it has received every message (one MCM.request plus one MCM.response per TV). 

After the maneuver concludes, it is verified by both the SV and all TVs identified by it, issuing a VEE with the Ledger module linking the previous block and the informational flag set to determine if the maneuver was executed successfully or not. A second block is created using all the exchanged VEE-extended messages in this second VsE. These VEEs can be issued on top of any ITS message using the passive mode. The order of the messages in the block is in an ascending manner, according to the station IDs of the maneuver participants. A station can only create a block if it has received every message (one extended ITS message per SV/TV). 

A representation of a cooperative maneuver event can be visualized in Figure \ref{fig:maneuver-sequence}.

\begin{figure}
    \centering
    \includegraphics[width=1\linewidth]{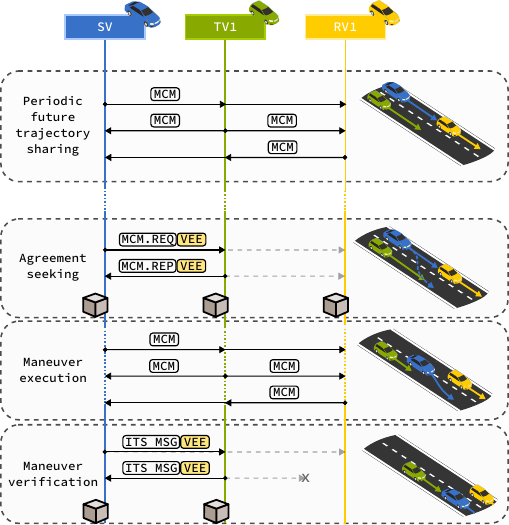}
    \caption{Cooperative maneuver SP example based on ETSI's Maneuver Coordination protocol. Vehicles use VEE to seamlessly record the negotiated maneuver data in a localchain block during the agreement seeking stage. After maneuver execution, vehicles verify if the maneuver was concluded successfully, generating another localchain block. Subsequent optional consensus processes to record additional maneuver data are not shown.}
    \label{fig:maneuver-sequence}
\end{figure}

The Token module (in combination with the Ledger module) can also be employed to provide a transactional feature between vehicles for maneuver advantages or simply as a gamification feature. Here, the promise mechanism of the Token module is employed. 
In the first VsE, the MCM.request includes, in addition to the Ledger container, the Token container defining the amount of tokens to transfer, equally distributed, to the TVs. 
The TVs willing to receive tokens issue the VEE of the MCM.response, with the Token container also included, containing the output address to which tokens shall be transferred. 
In the second VsE, both the SV and TVs verify if the negotiated maneuver concluded successfully through the previously mentioned extended ITS messages.
If they all agree then the token transfer is authorized. Eventually, this SP definition can change regarding which stations participate in the maneuver verification in order to exclude the SV in this VsE. This considers possible malicious SVs that promise tokens but then unverify successfully performed maneuvers. In this case, the total number of maneuver participants should be larger than two.

The Consensus module, combined with the Ledger module, can also be employed to increase the amount of data safeguarded in the ledger after the first two VsEs. The additional data to be recorded in the ledger could be messages exchanged deemed important during maneuver execution regardless of whether it finished successfully or failed. As in the previous VsEs, any data proposed here must be correctly signed. 
As mentioned, this  can include previous ITS messages, or the VEE-extended ITS message issued by the consensus process initiator containing the proposed data.
This consensus process could be initiated by any of the maneuver participants, other relevant vehicles, or onlookers like RSUs. A BFT consensus algorithm can be advised to account for possibly misbehaving users who may incorrectly disagree with the data proposed (e.g., data is proof against them). This consensus process runs after the maneuver verification, it is associated with an additional VsE, and creates a new block (not pictured in Figure \ref{fig:maneuver-sequence}).

The data constituting each VsE and the messages exchanged to create the VsEs in the Maneuver SP are summarized in Table \ref{table:maneuver-sp}. Optional modules and VsEs are inside parentheses.

\begin{table}[]
\centering
\caption{Maneuver SP VE format and associated messages.}
\begin{adjustbox}{width=\columnwidth,center}
\begin{tabular}{lll}
\hline \hline
VsEs &  Data & Messages exchanged \\ \hline
\multirow{2}{*}{\begin{tabular}[c]{@{}l@{}}VsE$_1$: \\ Agreement seeking\end{tabular}} & \multirow{2}{*}{\begin{tabular}[c]{@{}l@{}}Same as \\ exchanged messages\end{tabular}} & \multirow{2}{*}{\begin{tabular}[c]{@{}l@{}}MCM.request, MCM.responses \\ extended w/ VEE Ledger + (Token)\end{tabular}}   \\  & &   \\ \hline
\begin{tabular}[c]{@{}l@{}}VsE$_2$: \\ Maneuver verification\end{tabular}                  & \begin{tabular}[c]{@{}l@{}}Same as \\ exchanged messages\end{tabular}           & \begin{tabular}[c]{@{}l@{}}ITS messages \\ extended w/ VEE Ledger + (Token)\end{tabular}    \\ \hline
\begin{tabular}[c]{@{}l@{}}(VsE$_n$):\\ Proposed extra data\end{tabular}  & {\begin{tabular}[c]{@{}l@{}}ITS messages  and/or \\ other signed data\end{tabular}}
         & \begin{tabular}[c]{@{}l@{}}ITS messages \\ extended w/ VEE Ledger + Consensus\end{tabular} 
\\ \hline \hline
\end{tabular}
\end{adjustbox}
\label{table:maneuver-sp}
\end{table}

\subsection{View SP}\label{sec:view-sp}

A road user's perspective or vision over the vehicular environment is mostly composed of the information a) perceived by its own sensors or b) collected through vehicular communications from other road users. Such information is usually stored in conceptual data stores such as a Local Dynamic Map (LDM) \cite{ldm-standard}, combining the data of several vehicular protocols. Here we refer to a road user's perspective over its surroundings and system state as the \textit{view}.

The View SP provides a way for a station to verify its current view by harnessing the consensus feature of VEE. A road user can query other road users with questions about the current system state through a distributed verification methodology, and update its view accordingly.

This SP employs the Consensus module. Optionally, it can also include the Ledger module. Each VE of this SP contains at least one VsE, associated with the consensus process of the view update. More VsEs can be included with additional information or related view updates, each associated with a consensus process.

The initiator of the consensus process, proposing the view update, is known as the proposer. The proposer determines the node membership of the consensus process based on the proposal's characteristics and their relationships with the various road users as perceived by the proposer. For example, the proposer may prioritize selecting vehicles that likely have visual information about the proposal, as this could be beneficial in the decision-making process. Eventually, mechanisms like proof-of-eligibility \cite{liu19} or proof-of-relevance \cite{cao08} could be employed. These mechanisms aim at providing a method of only allowing road users which may have knowledge or are relevant to a vehicular event, to participate in the consensus process. 

If the same instance or event triggers several vehicles, and each starting a consensus process around the same time, the first process (with the lowest timestamp value in the first message associated with the process) takes precedence.

If the view update information is deemed important, the Ledger module can be used to increase the data security of this information. The proposer, in addition to proposing the data associated with the view update, also proposes the ledger-related information. If the consensus process finishes successfully, a localchain block is created with the message containing the proposed data, which the proposer used to initiate the consensus process. If previous exchanged ITS messages are referenced in the proposed data, these too should be recorded. It is advised to record only signed data in order to add authenticity and safeguard against malicious data forging.
Nearby stations listening to but not participating in the consensus process can also append the agreed block to their localchain instance.

The messages used to create the VsEs of the View SP are summarized in Table \ref{table:view-sp}. Optional modules and VsEs are inside parentheses.

\begin{table}[h]
\centering
\caption{View SP VE format and associated messages.}
\begin{adjustbox}{width=\columnwidth,center}
\begin{tabular}{lll}
\hline \hline
VsEs &  Data & Messages exchanged \\ \hline
\begin{tabular}[c]{@{}l@{}}(VsE$_1$):\\ Proposed view\end{tabular}  & {\begin{tabular}[c]{@{}l@{}}ITS messages and/or \\  other signed data\end{tabular}}
         & \begin{tabular}[c]{@{}l@{}}ITS messages extended w/ \\  VEE Consensus + (Ledger) \end{tabular} \\ \hline
\begin{tabular}[c]{@{}l@{}}(VsE$_n$):\\ Additional proposed data\end{tabular}  & {\begin{tabular}[c]{@{}l@{}}ITS messages and/or \\  other signed data\end{tabular}}
         & \begin{tabular}[c]{@{}l@{}}ITS messages extended w/  \\ VEE Consensus + (Ledger)\end{tabular} 
\\ \hline \hline
\end{tabular}
\end{adjustbox}
\label{table:view-sp}
\end{table}

Figure \ref{fig:view-sequence} presents an example of the View SP process. Here, a vehicle issues a DENM message stating the presence of an animal on the road. Another vehicle does not see the animal's presence at the reported location and remains dubious about the report. It then asks other vehicles which can possibly verify its presence. Through a consensus process, the quorum reaches a decision - the animal does not exist - and these issue a DENM negating the original DENM. Although not shown, the Ledger module could be included in the message exchange by V1, creating a block with the information of the consensus process proposal at the end of the consensus process. After the process, the quorum could also report the possible misbehaving vehicle, which reported the non-existent animal to a misbehavior authority (MA) \cite{ma-standard}.

\begin{figure*}
    \centering
    \includegraphics[width=1\textwidth]{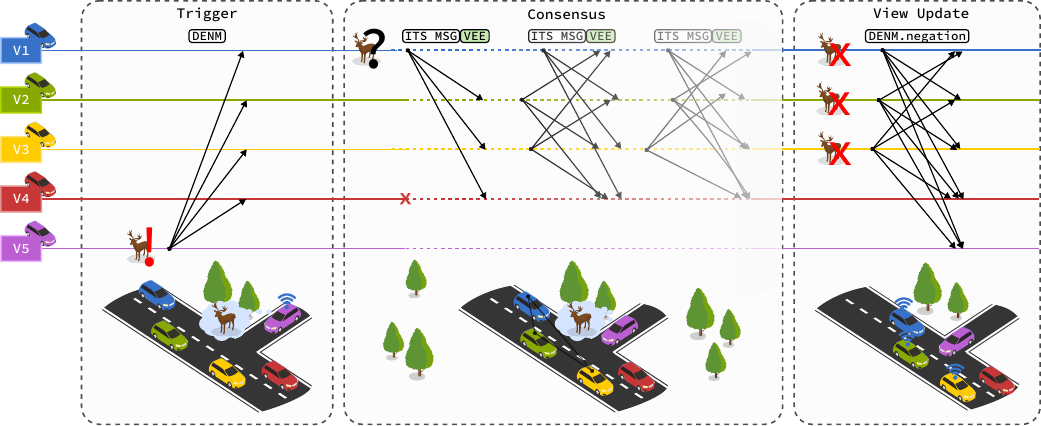}
    \caption{View SP example. V5 first issues a DENM.\textit{animalOnTheRoad}. V1 is doubtful and selects a group of vehicles that can possibly verify the situation. A fault-tolerant consensus algorithm is used using standard ITS messages as a transport medium. V4 is faulty and does not participate in the process. After reaching a consensus, the quorum finds no evidence of an animal and issues a negation DENM.}
    \label{fig:view-sequence}
\end{figure*}

\subsection{Tolling SP}\label{sec:tolling-sp}

Currently, the only standardized protocol for V2X-based toll usage is defined by SAE \cite{sae-j3217}. This protocol uses RSUs to periodically broadcast Toll Advertisement Messages (TAMs) encompassing information about the toll operator, the toll zone, and fees for each toll. Whenever a vehicle finds itself inside the toll-defined zone, it issues a Toll Usage Message (TUM) encompassing client-related information, encrypted (ECIES encryption scheme) using the RSU's provided public encryption key. The RSU then processes the TUM and issues an encrypted Toll Usage Message Acknowledgement (TUMack), confirming or rejecting the vehicle's TUM.

Here, this SP adds a transactional feature based on cryptocurrency through VEE's Token module. One VsE is associated with this SP, related to the performed transaction. The Token module's offer mechanism is used.
Both the TUM and TUMack are extended with the VEE. 
The messages used to create the single VsE of the Tolling SP are summarized in Table \ref{table:tolling-sp}. 
A representation of the toll usage protocol is presented in Figure \ref{fig:toll-sequence}.

\begin{table}[h]
\centering
\caption{Tolling SP VE format and associated messages.}
\begin{adjustbox}{width=\columnwidth,center}
\begin{tabular}{lll}
\hline \hline
VsEs &  Data & Messages exchanged \\ \hline
\begin{tabular}[c]{@{}l@{}}(VsE$_1$):\\ Toll payment\end{tabular}  & \begin{tabular}[c]{@{}l@{}}Same as \\ exchanged messages\end{tabular}
& \begin{tabular}[c]{@{}l@{}}TUM, TUMack \\ extended w/ VEE Token \end{tabular} 
\\ \hline \hline
\end{tabular}
\end{adjustbox}
\label{table:tolling-sp}
\end{table}

\begin{figure}
    \centering
    \includegraphics[width=1\linewidth]{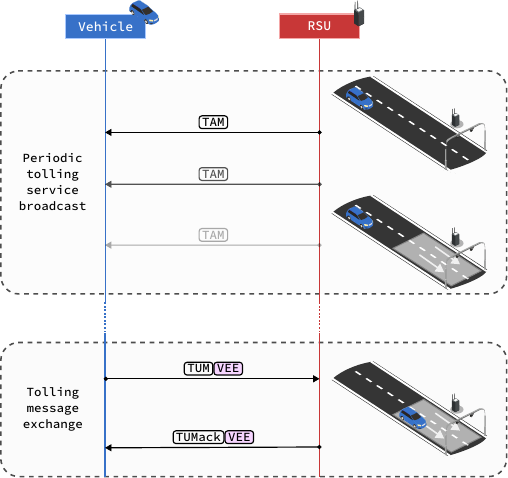}
    \caption{V2X-based fee collection SP. VEE is added to provide a transactional feature among stations, enabling the vehicle to perform the due payment on the spot.}
    \label{fig:toll-sequence}
\end{figure}

Given that the TUM contains the client's identity, only the exchanged VEEs, without the ITS message data, can be extracted from the exchanged packets to preserve the client's privacy. The VEE's Token module signed data allowing the transfer of tokens can then be forwarded to the blockchain associated with the transaction by either the OBU or RSU.

Given that the TUM-TUMack exchange can fail, it is normal for the OBU to keep retransmitting the TUM until it receives the TUMack from the RSU. Due to the possible retransmissions, it is necessary for the signed data included in the VEE not to change across retransmissions and not to trigger possible different transactions. If the signed data is required to change, the blockchain should implement a transaction rate limit mechanism, either through, for example, not accepting different transactions with the same ID/nonce, or if timestamps are included, not accepting transactions coming from the same sender spaced less than some time delay apart. 

The output address, that is, the address owned by the tolling service provider, for which the OBU must transfer the due fee, must be either advertised in the TAM or previously communicated/fetched through a secure ``offline'' method. If there is a fear that the RSU could be exploited by a malicious actor, e.g., installed in an easily reachable location, the service provider should prefer the latter method in order not to risk misappropriated funds.

\section{Simulation Experiments}\label{sec:simulations}

The three previously described SPs were implemented and analyzed using hardware-in-the-loop (HiL) based simulations to evaluate the performance of VEE and its modules. For each SP, only the non-optional VEE modules were employed, ultimately having simulations focused on the performance of each specific module.
The simulation setups are first described, followed by the obtained measurement results.

\subsection{Setups}

ETSI ITS stations were used in these HiL-based simulation setups. Instead of providing these platforms with the usual global navigation satellite system (GNSS) signal, the stations' positions and dynamics were simulated by feeding the protocol stack with a path composed of WGS 84 coordinates. Each platform calculates its current position, speed, and planned trajectories as the simulations run. Upon reaching the end of the path, the position is reset to its beginning. The setups are run in real time, and no other features nor operations besides the vehicle dynamics are simulated.

The required changes to the protocol stack were made as mentioned in Section \ref{sec:system-model}. That is, the ITS service at the Facilities layer, after generating each message, issues a request for any extension to the Applications service before forwarding the extended (or not extended) message to the lower layers of the stack.

A highway lane merging scenario was employed in the simulations. Two paths were used: the highway path and the on-ramp path. Only the Cooperative Awareness (CA) basic service, responsible for generating CAMs, was used as a UP to transmit passive operation mode VEEs. The distribution of the CAM generation period is shown in Figure \ref{fig:up-period}. The average CAM transmission periods are $\mu=881.48$ ms for the highway path and $\mu=591.22$ ms for the on-ramp path. The on-ramp path has a higher frequency of CAM transmission due to the higher curvature of the road. Throughout the simulations, vehicles are systematically assigned to one of the two paths, evenly distributed, and always within communications range.

\begin{figure}
    \centering
    \includegraphics[width=1\linewidth]{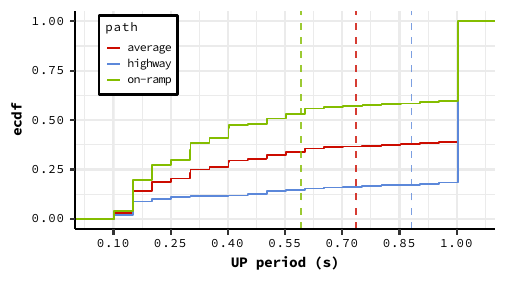}
    \caption{Period $T$ distributions of the underlying protocol (UP) used in the VEP passive operation mode. The UP is the CA basic service (CAM generator). The CAM period distributions for the two used paths, and their average, in the used highway lane merging scenario, are shown. Vertical lines show average values.}
    \label{fig:up-period}
\end{figure}

To analyze various performance parameters, three setups (A, B, and C) were used to analyze the SPs. A representation of these are presented in Figure \ref{fig:setup}.

\begin{figure}
    \centering
    \includegraphics[width=1\linewidth]{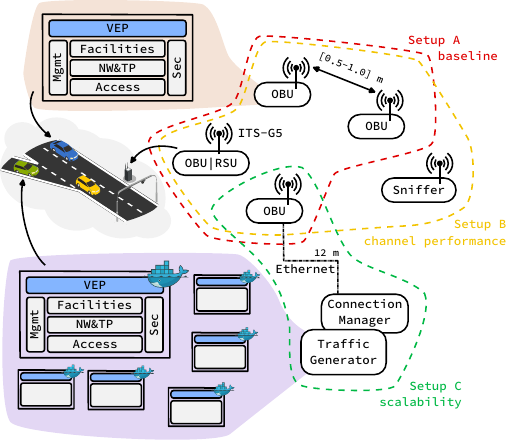}
    \caption{Three setups (A, B, and C) used in the simulation experiments.}
    \label{fig:setup}
\end{figure}

Setup A (baseline) focuses on analyzing the performance of the proposed design for a small number of stations, with the OBUs/RSUs running on common ITS station hardware. The time-based performance delays of the communication protocols and SPs are measured.

Setup B (channel performance) focuses on the overhead of the VEE on the ITS-G5 communications channel through measurements of the channel busy ratio (CBR) using a configuration similar to setup A. Here however, the employed communication protocols were modified to increase the VEE traffic rate in order to better study its overhead in higher saturation scenarios. The CBR is calculated by determining the intervals when the communications channel is busy, divided by the total time taken for the measurement. A Linux-based network sniffer with an Atheros AR9280 chipset was used to measure this value. It was tuned into the used ITS-G5 communications channel. The measurements of the CBR were taken by analyzing the busy time output of the \verb|iw| utility. The CBR values of simulation runs employing the SPs (providing the VEE-extended traffic) were compared with the CBRs of simulation runs not employing the SPs (no extended messages).
Decentralized Congestion Control (DCC), an algorithm to schedule messages in high traffic scenarios, was not in effect as the CBR remained far below the 30\% threshold \cite{dcc} in all simulations due to hardware limitations (number of units) to generate higher traffic.

Setup C (scalability) focuses on the scalability of the proposed design for a variable number of stations, using one OBU running on common hardware mixed in with a traffic generation system (TGS) running Docker containers simulating vehicles  ($[$3-27$]$). 
Due to the difference in performance between the OBU and the TGS, the later was not employed in the previous setups in order not to manipulate the measured delays. 
Therefore, this setup was only applied to the View SP since it is the SP which mostly uses the VEE passive operation mode.
The delay of SPs relying on this operation mode is less dependent on the computational performance and more on the frequency of the UP, as assumed in \circled{1}. This setup can also control the packet delivery ratio (PDR) among stations, by discarding the transmission of randomly chosen messages to randomly selected stations. The PDR (\%) is the ratio,

\begin{equation}
    PDR = \frac{N_{rx}}{N_{tx}} \times 100
\end{equation}
where $N_{rx}$ is the number of received packets, and $N_{tx}$ is the number of transmitted packets. A simple Bernoulli distribution was used, where packets on the reception side had a certain chance of being ignored. This laboratorial setup allows for a more meticulous analysis of the behavior of the proposed system under specific packet loss values when compared to a more dynamic setup considering the various distances among vehicles or other factors that may affect the PDR.

Only the View SP was tested under such circumstances (setup C) since it is the only protocol that implements packet retransmission (through PBFT) in its consensus process. In the other SPs, if there is any packet loss, the associated UP processes simply restart from the beginning due to their lower message complexity (request-response). As such, the overall average delay of such processes can be easily approximated by multiplying the average delay of the optimal scenario by the number of times it retries. 

While setups A and B solely use ITS-G5 wireless communications with the stations located less than 1 meter apart from each other, setup C also employs Ethernet-based communications with a total cable length of a dozen meters. In the latter, the OBU and all the virtualized OBUs are connected to a connection manager, distributing the packets among them and controlling the PDR. This connection manager runs on the same machine as the traffic generation system, using the same emulation framework employed in \cite{verco}.

Hardware-wise, the ETSI ITS stations (the OBUs and RSU) comprise a host PC Engines APU3D4 with a 4-core AMD GX-412TC CPU @ 1.4 GHz with 4GB RAM, running Arch Linux x86-64. The ITS-G5 protocol stack and all associated packet-handling operations are run on these hosts. Several services (C programs) compose the protocol stack with communications established between them using ZeroMQ inter-process communication (IPC) connections. To these hosts, ITS-G5 communications are either provided by a u-blox VERA-P1 module or a Compex WLE200NX module. The communications bitrate is set to the ITS-G5 default of 6 Mbps. Every exchanged message is signed using the algorithm ECDSA on the 256-bit Brainpool R1 curve. 
The network sniffer shares the same hardware specs as the ETSI ITS stations.
For the traffic generation system, both the virtualized OBUs (Docker containers) and connection manager are running in a 8-core Intel Core i9-9900K CPU @ 5.0 GHz with 32 GB RAM running Arch Linux x86-64.

\subsubsection{Maneuver SP}
A maneuver coordination protocol following current standardization efforts was employed, as described in Section \ref{sec:maneuver-sp}. Since this work focuses on the performance of the communications, the vehicle dynamics were largely simplified.

Vehicles coordinate maneuvers in the highway lane merging scenario. A slightly modified version of the under-standardization MCMs was used to only include relevant fields. The vehicles periodically broadcast MCMs, including their planned trajectory over the future 5 seconds, driving at a default speed of 90 km$^{-1}$. 
In this setup, upon a trajectory overlap (collision risk) with other TV(s), the vehicle that first detected it (the SV) issues a VEE-extended MCM.request for a trajectory equal to its currently planned/desired trajectory. TVs then respond to the SV, accepting the request by issuing a VEE-extended MCM.response. 
After the response, depending on which vehicle (the SV or TV) reaches the intersection point faster, the TVs either accelerate or decelerate their current speed by 10\%. After reaching the (acc/de)celerated target speed, the vehicle (ac/de)celerates back to the default speed. 
A timeout mechanism is implemented to safeguard against communication failures, starting when the SV issues a MCM.request. This mechanism ensures that the SV does not indefinitely wait for maneuver responses.

The first VsE ledger block is created upon receiving the MCM.request message and all MCM.responses. Five seconds later, the SV and TVs verify the execution of the maneuver by issuing VEE-extended ITS messages (here CAMs), and create the second VsE ledger block. % maybe add more info on how the maneuver is verified, what data is in the VEEs

Upon the simulation's initial state, each vehicle stores only one block, the genesis block, in their system's singular localchain instance. As the simulation progresses, vehicles gradually build up their localchain by adding blocks as multiple maneuvers occur.

Two out of the four OBUs were running in setup A, each following a different path of the lane merging scenario. In setup B, four OBUs were running, constantly performing the maneuver coordination protocol, i.e., the maneuver request-response, regardless if a trajectory overlap is detected or not. In these setups, the OBUs exchange CAMs and MCMs. 

\subsubsection{View SP}
These simulations focus on the performance of the consensus algorithm, which should account for the large majority of the overall delay of the View SP. We employ the Practical Byzantine Fault Tolerance (PBFT) consensus algorithm here, one of the most popular consensus algorithms. The trigger is employed using a simple time-based function, which triggers the consensus process every $P$ milliseconds. This function cannot initiate a new consensus process if an ongoing consensus process exists. In other words, at most one consensus process can be running at any given time. Across simulations runs, only one specific station could initiate the consensus process, in order to further strengthen this rule. The proposed data by this station/primary node is a buffer of random data, 32 bytes in length.

All nodes are honest and accept every proposal. No faults were introduced except for the random packet loss in setup C. A retransmission timeout value of $\tau_d=2200$ milliseconds was used. This value was chosen as twice the $T_{max}$ (1000 ms), plus a leeway for computational time. A maximum limit of 5 retransmissions per station per consensus process, and an upper limit of 30 seconds per consensus process were also set to safeguard against unreasonable lengthy processes. Failing to reach a consensus, i.e. receive $2f$ commit messages before the mentioned limits are surpassed, the station considers the process a failed process. If a consensus is reached, the PBFT process is deemed successful.

All $n$ stations participate in every consensus process. For setups A and B, $n=4$ OBUs were used. In setup A and C, the trigger period $P=30000$ ms (30 seconds), while in setup B it was varied in order to evaluate the CBR and protocol overhead variation. For setup C, $n=\{3,6,...,27\}$ virtualized OBUs were used, plus one running on actual OBU hardware (primary node). Results were measured from the viewpoint of the latter.  In these setups, the OBUs exchange only CAMs.

\subsubsection{Tolling SP}
To analyze the Tolling SP, an RSU serving as a service provider and vehicles acting as clients were used in the simulations. The RSU is constantly broadcasting toll advertisement messages containing information about a single toll zone. Upon crossing the toll zone, the vehicle issues a VEE-extended TUM with the Token container comprising the offer mechanism. This Token container includes a VERCO transaction, containing an ECDSA 256-bit signature (same level of cryptography as used in Bitcoin and Ethereum transactions) of the concatenated data of the nonce and transaction data (tokens offered). Appended to the signature, a certificate associated with a pseudonym holding the tokens in the VERCO ecosystem is also included. Upon receiving the TUM, the RSU verifies and responds with a VEE-extended TUMack including the Token module, with a confirmation code of the proposed cryptocurrency-based transaction, and also with a signature appended with the respective certificate. Both the TUM and TUMack are signed and encrypted at a lower level of the protocol stack using the ECDSA and ECIES protocols, also using 256-bit security.

While SAE J3217 \cite{sae-j3217} was defined to operate over the IEEE WAVE protocol stack, here it was implemented over the ETSI ITS protocol stack with different routing/networking protocols and slightly modified toll messages though with equivalent features. Instead of using the WAVE Service Advertisement (WSA) \cite{ieee1609.3} protocol to disseminate TAMs and advertise the tolling zones, the Services Announcement (SA) \cite{saem-standard} protocol was used to disseminate SAEMs containing this information. Instead of using the WAVE Short Messaging Protocol (WSMP)  \cite{ieee1609.3} to provide networking and transport features to the TUM/TUMack messages, the GeoNetworking (GN) \cite{gn-standard} protocol and the Basic Transport Protocol (BTP) \cite{btp-standard} were used to provide these features.

For setup A, only one OBU and one RSU were running. A single path was utilized in this setup, along which the vehicle takes one minute to complete a full run. The toll zone was positioned at the midpoint of the path. 
For setup B, one RSU and three OBUs were running. Here, the OBUs triggered the toll usage protocol by issuing at TUM at fixed time intervals, regardless if they were inside the toll zone or not. The RSU can handle several tolling requests simultaneously.
In these setups, the stations exchange CAMs, SAEMs (RSU), TUMs (OBUs) and TUMacks (RSU). 

% maybe add table summarizing all the tests configurations
\subsection{Results}

The delay and channel load performance of the different SPs in the previously described setup are presented below. Results are shown in milliseconds (ms) or in hertz (Hz). The use of  empirical cumulative distribution functions (eCDFs) is largely employed to represent the measured results in setups A and C (performance delays). For setup B (channel load), measured results are simply plotted over the associated protocol trigger rate values.

\subsubsection{Maneuver SP}

The measured delays of the Maneuver SP for 2 vehicles are shown in Figure \ref{fig:maneuver-delay} (setup A). The delays of both VsEs (maneuver coordination and maneuver verification) are presented. The theoretical expected delay of the verification is also shown, calculated using Equation \ref{eq:delay-ksrounds} for $k=1$ sub-round, and $m=2$ messages exchanged and $g=2$ required messages, using the $W$ distribution. 

Four delays were measured: the agreement seeking/maneuver coordination message exchange delay (``coordination'' on Fig. \ref{fig:maneuver-delay}) , the maneuver verification delay (``verification''), the VEE overhead delay (``vee''), and localchain block creation delay (``block forge''). The coordination delay, without the VEE overhead, is also shown (``req-rep''). 
The maneuver coordination message exchange delay was measured starting at the instant an SV detects a trajectory overlap with a received MCM and generates a VEE-extended MCM.request, up until it receives the VEE-extended MCM.response message issued by the TV. This delay includes the overhead introduced by the VEE, both on the SV and TV sides. The maneuver verification delay starts 5 seconds after the maneuver coordination block creation, and finishes when all OBUs participating in the maneuver (2 here) issue a VEE-extended CAM to verify the maneuver. This delay includes the overhead introduced by the VEE by just one vehicle since this process is concurrent among all vehicles. The localchain block creation delay is the time it takes for the station to create a block (hash using SHA-256 of all the coordination/verification messages) and add it to its own instance of the localchain. This process happens after both the agreement seeking and maneuver verification and does not directly affect the safety-critical delay of the related UPs.

\begin{figure*}
    \centering
    \includegraphics[width=1\linewidth]{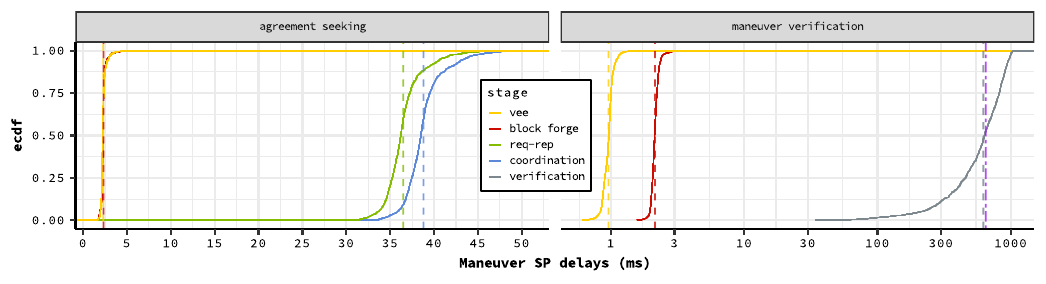}
    \caption{Maneuver SP measured delays for 2 vehicles coordinating $N=1000$ maneuvers (setup A). Vertical lines represent average values. A logarithmic scale is used for the maneuver verification delays. The theoretical expected delay of the verification is also plotted (vertical purple line).}
    \label{fig:maneuver-delay}
\end{figure*}

The average maneuver coordination delay was $\mu = 38.79$ ms with a standard deviation of $\sigma = 2.36$ ms. This value includes the delay introduced by the VEEs on the MCM messages which was $\mu = 2.31$ ms with $\sigma = 0.35$ ms, including both the request and response. 
This corresponds to a time overhead of $+2.97\%$ introduced by VEE in the maneuver coordination protocol. 
For the maneuver verification, the average verification delay was $\mu=622.03$ ms with $\sigma= 235.15$ ms. 
This delay includes the delay introduced by the VEE on the concurrent CAM messages which was $\mu = 0.96$ ms with $\sigma = 0.09$ ms. The theoretical expected delay of the verification is $E(W^2_{(2)})= 648.48$ ms, a little higher when compared to the measured results due to the skewing of the UP period $T$ (and $W$) towards lower values as the vehicles coordinate maneuvers, changing speeds thus leading to an higher CAM frequency. The average block creation delay for both VsEs was $\mu=2.21$ ms with $\sigma=0.43$ ms.

To analyze the performance of the ITS-G5 communications channel, setup B was used with 4 OBUs running, where one was constantly requesting maneuvers from the other three. MCM.requests were issued with a period ranging [100-20000] ms. 
The observed packet size ranged between [92-692] bytes for the normal traffic, and [92-735] bytes for the traffic with the Maneuver SP. 
The sniffer measured the CBR over a timeframe of $2500$ seconds for each period value and for each traffic type.
The radio channel performance results are shown in Figure \ref{fig:ledger-cbr}. 
Results show a lower average VEE overhead for larger MCM.request trigger periods as the system handles a less amount of maneuver coordination processes, equivalent to less extended MCM.requests/responses being exchanged. For smaller MCM.request trigger periods, the exchanged VEE-extended MCM.requests/responses and verifying CAMs dominate the network traffic, transpiring an higher average VEE overhead.
The VEE overhead varied from  $+7.75\%$ to $+0.21\%$. These results show the low impact of the VEE for data security in cooperative maneuvers, even in scenarios requiring an high abnormal amount of coordination processes involving several vehicles.

\begin{figure}
    \centering
    \includegraphics[width=1\linewidth]{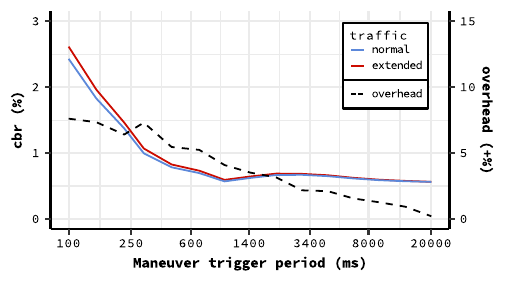}
    \caption{Channel load of $n=4$ vehicles performing coordinated maneuvers. All vehicles participate in every maneuver. The maneuver coordination protocol is triggered periodically ($x$-axis). Both normal (non-extended) and extended (Maneuver SP) traffic is shown. The Maneuver SP overhead scale follows the right $y$-axis. Each measured value corresponds to an observation time of 2500 seconds.}
    \label{fig:ledger-cbr}
\end{figure}

%\begin{table}[]
%\caption{Maneuver SP overhead (setup B).}
%\centering
%\begin{tabular}{lcc}
%Traffic  & Normal & VEE-extended \\ \hline  
%\# of stations ($n$) & \multicolumn{2}{c}{4} \vspace{1mm}\\
%Communication medium & \multicolumn{2}{c}{ITS-G5 @ 5.9 GHz} \vspace{1mm} \\ 
%Bitrate & \multicolumn{2}{c}{6 Mbps} \vspace{1mm}\\
%Packet loss & \multicolumn{2}{c}{0\% (none observed)} \vspace{1mm}\\
%MCM.request trigger period & \multicolumn{2}{c}{100 ms} \vspace{1mm}\\
%Vehicles per maneuver & \multicolumn{2}{c}{4 (1 SV, 3 TVs)} \vspace{1mm}\\
%Packet length (bytes) & [92-692] &   [92-735]   \vspace{1mm}\\
%CBR & 1.2708\% & \begin{tabular}[t]{@{}c@{}}1.4008\%\\ (+10.23\%) \end{tabular}  
%\end{tabular}
%\label{table:maneuver-cbr}
%\end{table}

\subsubsection{View SP}
The measured delays of the PBFT-based View SP for $N=1000$ processes and $n=4$ stations are shown in Figure \ref{fig:pbft-delay} (setup A). The average expected delay calculated from Equation \ref{eq:passive-delay} is also plotted. 
The overhead of VEE is not shown as it is minimal ($\approx 1$ ms) when compared with the measured average PBFT process delays.
The PBFT process delay measurement starts when the station initiates (primary node) the consensus process and finishes when it has received all the necessary commit messages (2, plus its own). 
No packet loss was observed.

\begin{figure}
    \centering
    \includegraphics[width=1\linewidth]{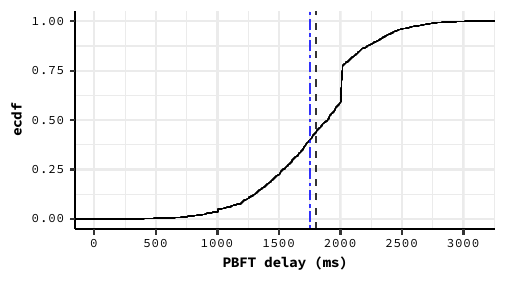}
    \caption{PBFT consensus delay of View SP for $n=4$ vehicles and $N=1000$ PBFT processes (setup A). Vertical lines represent measured average (black) and theoretical expected (blue) values.}
    \label{fig:pbft-delay}
\end{figure}

On average, the total average delay it took for stations to reach a consensus was $\mu=1804.27$ ms, with a standard deviation of $\sigma = 418.22$ ms. The expected theoretical delay is $E(\tau_p) = 1754.11$ ms, a bit smaller than the measured delay due to assumption \circled{1}.
%The measured average was slightly lower than the theoretical expected average of $1843.49$ ms due to assumption \circled{2} which ignores the fact that the first message of the consensus process can be sent before $T_{min}$. 

As expected, given the average delay of the UP, the total process average delay is too large for safety-critical applications in vehicular scenarios. Even if the available underlying traffic is tripled (e.g., use also MCMs and CPMs for consensus) and the expected average consensus delay reach $\approx 600$ ms, this value would still be considered too large for these type of applications.

For the performance of the ITS-G5 channel, setup B was similar to setup A, although with a varying PBFT trigger period ranging [500-30000] ms. 
The observed packet size ranged between [92-553] bytes for the normal traffic, and [92-619] bytes for the traffic with the View SP.
The sniffer measured the CBR over a timeframe of $2500$ seconds for each trigger period value for simulations with the View SP traffic. For the normal traffic, the CBR was measured just once, over a timeframe of also $2500$ seconds.
The radio channel performance results are shown in Figure \ref{fig:pbft-cbr}. 

Results show a lower average VEE overhead for larger PBFT trigger periods as the system handles a less amount of PBFT processes. The overhead remains fairly constant for PBFT trigger period below the previously average delay of $\mu=1804.27$ ms, starting to decrease around this threshold. Below this value, the vast majority of exchanged CAMs are extended with the VEE as more PBFT processes, on average, are created before they are finalized. 

The VEE overhead varied from $+3.73\%$ to $+0.45\%$, demonstrating the lightweightness of the proposed design. 
Considering the previously measured delays, but the low overhead of VEE in the communications channel, the proposed design offers a viable approach for road users to establish consensus on non-safety-critical vehicular applications.

\begin{figure}
    \centering
    \includegraphics[width=1\linewidth]{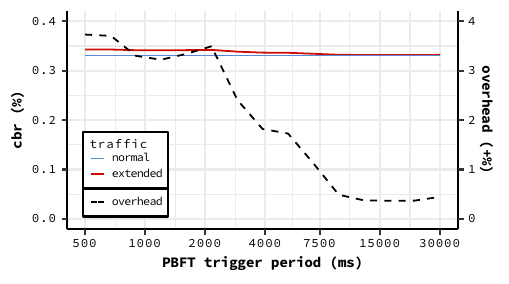}
    \caption{Channel load of $n=4$ vehicles running the PBFT algorithm. All vehicles participate in every process. The process is triggered periodically ($x$-axis). Both normal (non-extended) and extended (View SP) traffic is shown. The View SP overhead scale follows the right $y$-axis. Each measured value corresponds to an observation time of 2500 seconds.}
    \label{fig:pbft-cbr}
\end{figure}

%\begin{table}[]
%\caption{View SP overhead (setup B).}
%\centering
%\begin{tabular}{lcc}
%Traffic  & Normal & VEE-extended \\ \hline  
%\# of stations ($n$) & \multicolumn{2}{c}{4} \vspace{1mm}\\
%Communication medium & \multicolumn{2}{c}{ITS-G5 @ 5.9 GHz} \vspace{1mm} \\ 
%Bitrate & \multicolumn{2}{c}{6 Mbps} \vspace{1mm}\\
%Packet loss & \multicolumn{2}{c}{0\% (none observed)} \vspace{1mm}\\
%Packet length (bytes) & [92-553] &   [92-619]   \vspace{1mm}\\
%PBFT trigger period & - & 500 ms     \vspace{1mm}\\
%CBR & 0.3070\% & \begin{tabular}[t]{@{}c@{}}0.3189\%\\ (+3.88\%) \end{tabular}  
%\end{tabular}
%\label{table:view-cbr}
%\end{table}

\begin{figure}
    \centering
    \includegraphics[width=1\linewidth]{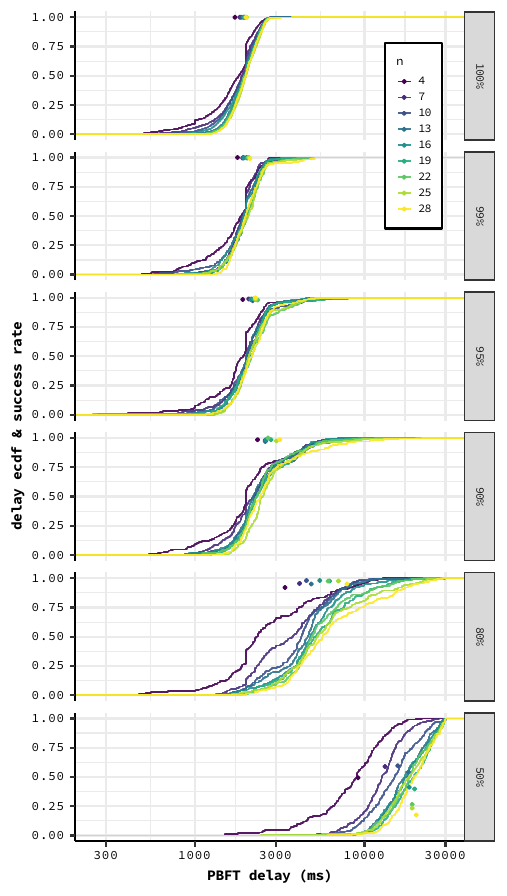}
    \caption{PBFT consensus delay of View SP for $n$ vehicles (setup C). Measured delays for different PDRs (\%) are presented. The ratio of successful processes is also shown (points, plotted over the average delays). A logarithmic scale for the delay is used.}
    \label{fig:pbft-s-delay}
\end{figure}
Regarding the scalability analysis, Figure \ref{fig:pbft-s-delay} provides the measured PBFT delays and successful rate for different numbers $n$ of participating stations and different PDRs (setup C). For each combination (each plotted eCDF) of $n$ vs. PDR, the delay of $N=200$ successful processes was measured. The successful rate is calculated using number of successful processes over the number of total processes. This value was also measured for $N=200$ total processes.
Results show that, as $n$ increases, so does the average delay due to the need for more messages to be exchanged. A higher number of stations in the system also increases the overall network traffic, lowering the computational performance of each station, increasing the time to generate and analyze the ITS messages and therefore increasing the overall PBFT delay. Regarding the PDR, PBFT performance remains stable for high enough PDRs ($\geq 90\%$) for different $n$ values, with performance degrading noticeably as the PDR decreases. About the success rate, it is reasonably high for most PDRs, 100\% for PDR = 100\%, and $\geq 87.5\%$ for PDR $\geq 80\%$. However for PDR = 50\%, the success rate is considerably lower, specially when considering larger $n$ values, where most processes fail. 

Without using VEP/VEE and the passive operation mode, the straightforward use of PBFT or higher complexity consensus algorithms may also not be feasible in safety-critical operations \cite{mxm-pbft}, especially when considering non-negligible packet loss scenarios. 

%more context is needed for this comparison with the previous paper

\subsubsection{Tolling SP}

The measured delays of the Tolling SP are shown in Figure \ref{fig:tolling-delay}. 
Three delays were measured: the total tolling message exchange delay (``total'' in Fig. \ref{fig:tolling-delay}), the OBU VEE overhead (``vee (obu)''), and the RSU VEE overhead (``vee (rsu)''). 
The delay of the TUM-TUMack exchange, without the VEE overheads is also provided (``tum-tumack'').
The total tolling message exchange delay was measured starting when the OBU detects it crosses the toll zone and generates a VEE-extended TUM until it receives the VEE-extended TUMack message issued by the RSU. The OBU VEE overhead is the time for the TUM to be extended with the VEE, where a transaction is issued (256-bit ECDSA signature). The RSU VEE overhead is the time for the TUMack to be extended with the VEE, where the transaction is verified (256-bit ECDSA signature verification). Both of these overheads contribute to the total tolling message exchange delay.

The average tolling message exchange delay was $\mu = 87.82$ ms with a standard deviation of $\sigma = 3.07$ ms. The OBU's VEE overhead was $\mu = 7.04$ ms with $\sigma = 0.79$ ms, and the RSU's VEE overhead was $\mu = 12.73$ ms with $\sigma = 1.26$ ms. 

\begin{figure}
    \centering
    \includegraphics[width=1\linewidth]{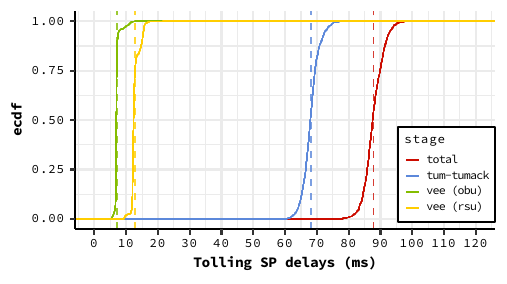}
    \caption{Tolling SP measured delays. Vertical lines represent average values.}
    \label{fig:tolling-delay}
\end{figure}

Ultimately, VEE introduced an average overhead of $\mu=19.77$ ms, equivalent to $+29\%$, to the TUM-TUMack exchange. This increase is the largest of the analyzed SPs mainly due to the additional computational power required for the cryptographic operations associated with issuing and verifying the VEE Token container signature.

%This outcome demonstrates the viability of employing VEE to enhance trading capability in V2X tolling services. Furthermore, the feasibility is reinforced by considering the non-safety-critical nature of these types of services.

To analyze ITS-G5 channel performance under heavier traffic, one RSU and three OBUs were used in setup B. The OBUs constantly issue TUMs at a rate of [0.05-6.67] Hz, regardless if they are crossing the tolling zone or not. This is equivalent to the RSU handling [0.15-20] toll requests per second.
The observed packet size ranged between [92-553] bytes for the normal traffic, and [92-756] bytes for the traffic with the Tolling SP.
The sniffer measured the CBR over a timeframe of $2500$ seconds for each TUM rate value and for each traffic type.
The measured CBRs for both normal and VEE-extended scenarios are shown on Figure \ref{fig:toll-cbr}.
Results show an higher average VEE overhead for larger request rates as the extended TUMs/TUMacks dominate the channel versus the exchanged CAMs and SAEMs. The VEE overhead varied from $+1.27\%$ to $+24.63\%$. This overhead is the largest of the analyzed SPs due to the larger VEE Token container with its added security (additional signature and associated certificate), increasing the TUM/TUMack packet size by $\approx 200$ bytes.

\begin{figure}
    \centering
    \includegraphics[width=1\linewidth]{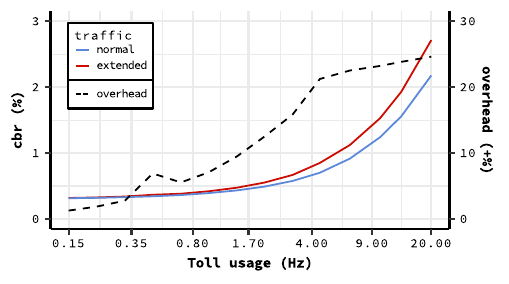}
    \caption{Channel load of one RSU and $n=3$ vehicles performing the toll usage protocol. Both normal (non-extended) and extended (Tolling SP) traffic is shown. The $x$-axis represents the number of toll requests (TUMs) the RSU receives per second in a logarithmic scale. The Tolling SP overhead scale follows the right $y$-axis. Each measured value corresponds to an observation time of 2500 seconds.}
    \label{fig:toll-cbr}
\end{figure}

%\begin{table}[]
%\caption{Tolling SP overhead (setup B).}
%\centering
%\begin{tabular}{lcc}
%Traffic  & Normal & Extended \\ \hline  
%\# of stations ($n$) & \multicolumn{2}{c}{4 (1 RSU, 3 OBUs)} \vspace{1mm}\\
%Communication medium & \multicolumn{2}{c}{ITS-G5 @ 5.9 GHz} \vspace{1mm} \\ 
%Bitrate & \multicolumn{2}{c}{6 Mbps} \vspace{1mm}\\
%Packet loss & \multicolumn{2}{c}{0\% (none observed)} \vspace{1mm}\\
%TUM trigger period & \multicolumn{2}{c}{100 ms} \vspace{1mm}\\
%Packet length (bytes) & [92-553] &   [92-756]   \vspace{1mm}\\
%%PBFT trigger period & - & 500 ms     \vspace{1mm}\\
%CBR & 3.1986\% & \begin{tabular}[t]{@{}c@{}}3.8833\%\\ (+21.41\%) \end{tabular}  
%\end{tabular}
%\label{table:toll-cbr}
%\end{table}

\section{Conclusions and Future Work}\label{sec:conclusion}

In this paper we presented VEE, a generalized extension for V2X messages with the purpose of increasing data security, consensus, and trading in vehicular networks. 
VEE was based on the concept of applicational extensions, a new mechanism for vehicular communications which can add or enhance features of vehicular networks, utilizing the existing V2X traffic as a transport medium. The mechanism of how applicational extensions could work in a typical protocol stack, focusing on the ETSI ITS protocol stack, was also explained. The VEE extension is built using independent entities named modules which provide the features of data security, consensus, and trading to V2X messages. Sub-protocols, each associated with a type of vehicular event, were also provided, acting here as examples of how VEE can be used for various vehicular scenarios.
The lightweightness of the proposed design, namely its impact on the computational performance and channel load, was validated using hardware-in-the-loop simulations.

Variations of the proposed system can be explored. The Ledger module employed localchains of the VERCO architecture due to their lightweightness and non-reliance on the infrastructure/RSUs, though eventually other types of blockchains/DLTs suitable for the opportunistic and mobile nature of vehicular networks could be employed. Regarding data security, other non blockchain/DLT based mechanisms could also be investigated.

A more thorough analysis of the Consensus module can be performed. We only tested its performance using the PBFT consensus protocol using the CAMs' UP. The use of other consensus protocols, fault tolerant or not, and a combination of several UPs can be analyzed to decrease the total delay of the consensus processes. Several membership selection algorithms can also be implemented.

The Token module can also be reformulated for better performance. It was the module with heaviest impact in both communication delay and channel performance in the analyzed setups, mainly due to the need for the additional security-features required by the employed cryptocurrency trading mechanism. A way to increase its performance would be to merge the cryptocurrency identity into the vehicular identity, approximately halving the cryptographic operations required. However this could raise some privacy issues, yet to be analyzed.

Ultimately, a multitude of modules or even new types of applicational extensions can be explored for a multitude of vehicular use-cases. 
As demonstrated here, these types of extensions can seamlessly add and enhance features of the vehicular network, while also minimally impacting the vehicular network.

\section*{Acknowledgment}

This work is supported by the European Union / Next Generation EU, through Programa de Recuperação e Resiliência (PRR) [Project Route 25 with Nr. C645463824-00000063].

\bibliographystyle{ieeetr}
\bibliography{refs}

\end{document}